\documentclass[11pt,letterpaper]{article}
\usepackage[margin=1in]{geometry}
\usepackage{amsfonts,amsmath,latexsym,amssymb,amsthm}
\usepackage[english]{babel}
\usepackage{graphicx}
\usepackage[dvipsnames]{xcolor}
\usepackage{xspace}
\usepackage{paralist}
\usepackage{enumitem}
\usepackage{subcaption}
\usepackage[hidelinks]{hyperref}
\usepackage{bm}
\usepackage{varwidth}
\usepackage[utf8]{inputenc}
\usepackage{lmodern}
\usepackage{microtype}
\usepackage{authblk}
\usepackage{tikz}
\usepackage{textcomp}
\usetikzlibrary{arrows.meta,tikzmark,positioning,decorations.markings}

\usepackage{algorithmic}
\usepackage[]{algorithm2e}
\title{Improved Approximation Guarantees for Shortest Superstrings using Cycle Classification by Overlap to Length Ratios}

\author[1]{Matthias Englert}
\author[2]{Nicolaos Matsakis}
\author[3]{Pavel Vesel\'y\thanks{Part of this work was done when the author was at the University of Warwick.
  	Partially supported by European Research Council grant ERC-2014-CoG 647557, by GA \v{C}R project 19-27871X, and by Center for Foundations of Modern Computer Science (Charles University project UNCE/SCI/004).}}

\affil[1]{University of Warwick, \texttt{M.Englert@warwick.ac.uk}}
\affil[2]{\texttt{nickmatsakis@gmail.com}}
\affil[3]{Charles University, \texttt{vesely@iuuk.mff.cuni.cz}}
\date{}

\def\OPT{\textsf{OPT}\xspace}
\def\ALG{\textsf{ALG}\xspace}
\def\GREEDY{\textsf{GREEDY}\xspace}
\def\MGREEDY{\textsf{MGREEDY}\xspace}
\def\TGREEDY{\textsf{TGREEDY}\xspace}

\def\CC{cycle cover\xspace}

\def\pref{\textsf{pref}\xspace}

\SetKwInput{KwData}{Input}
\SetKwInput{KwResult}{AND}

\def\C{\textsf{CC}}

\def\ov{\textsf{ov}\xspace}
\def\period{\textsf{period}\xspace}
\def\dist{\textsf{dist}\xspace}
\def\strings{\textsf{strings}\xspace}
\def\period{\textsf{period}\xspace}

\DeclareMathOperator*{\swap}{swap}

\newcounter{thm}
\numberwithin{thm}{section}
\theoremstyle{plain}
\newtheorem{theorem}[thm]{Theorem}
\newtheorem{corollary}[thm]{Corollary}
\newtheorem{lemma}[thm]{Lemma}
\newtheorem{fact}[thm]{Fact}

\newtheorem{definition}[thm]{Definition}
\newtheorem{observation}[thm]{Observation}
\theoremstyle{remark}
\newtheorem{remark}[thm]{Remark}

\newcommand\ignored[1]{}

 \newcounter{note}[section]

\begin{document}

\maketitle
\begin{abstract}
  In the Shortest Superstring problem, we are given a set of strings and we are asking for a common superstring, which has the minimum number of characters. The Shortest Superstring problem is NP-hard and several constant-factor approximation algorithms are known for it.
  Of particular interest is the \GREEDY algorithm, which repeatedly merges two strings of maximum overlap until a single string remains. The \GREEDY algorithm, being simpler than other well-performing approximation algorithms for this problem, has attracted attention since the 1980s and is commonly used in practical applications.

  Tarhio and Ukkonen~(TCS 1988) conjectured that \GREEDY gives a 2-approximation. In a seminal work, Blum, Jiang, Li, Tromp, and Yannakakis (STOC 1991) proved that the superstring computed by \GREEDY is a 4-approximation,
  and this upper bound was improved to 3.5 by Kaplan and Shafrir (IPL 2005).

  We show that the approximation guarantee of \GREEDY is at most $(13+\sqrt{57})/6 \approx 3.425$, making the first progress on this question since 2005.
  Furthermore, we prove that the Shortest Superstring can be approximated within a factor of $(37+\sqrt{57})/18\approx 2.475$, improving slightly upon the currently best $2\frac{11}{23}$-approximation algorithm by Mucha (SODA 2013).
\end{abstract}

\section{Introduction}
In the Shortest Superstring problem (SSP), we are given a set $S$ of strings over a finite alphabet, and we are asking for a string of minimum length, which contains each member of $S$ as a substring.
SSP has found important applications in various scientific domains~\cite{Gevezes2014}.
One of the early uses was for DNA sequencing \cite{leskBook,myers2016history}, where a DNA molecule consisting of four different nucleotides (Adenine, Thymine, Guanine, and Cytosine) is gradually assembled from DNA fragments, which can be viewed as an instance of SSP over a quaternary alphabet.
SSP can also arise in data compression \cite{storerBook}. Since information is represented by binary strings, we are asking for the minimum number of binary digits that can encode a larger set of strings.
Interestingly, SSP has been used to study how effectively viruses compress their genome by overlapping genes~\cite{ilie2006shortest}.

SSP is NP-hard,
even when the alphabet is binary \cite{garey}. Moreover, SSP is APX-hard \cite{blum} as it is not $(\frac{333}{332}-\epsilon)$-approximable for any constant $\epsilon>0$ unless $\mathrm{P}=\mathrm{NP}$ \cite{KarpinskiS11}.
There exists a plethora of constant-factor SSP approximation algorithms, the currently best of which has an approximation ratio upper bound of $2\frac{11}{23}=\frac{57}{23}\approx 2.478$ \cite{mucha}. Blum, Jiang, Li, Tromp, and Yannakakis~\cite{blum} showed that the \GREEDY algorithm, which repeatedly merges two strings of maximum overlap (breaking ties arbitrarily) until a single string remains, computes a 4-approximate superstring. Additionally, Blum et al.\ gave two simple variants of \GREEDY, namely \TGREEDY with approximation ratio at most~3 and \MGREEDY with ratio at most 4. A series of improved approximation algorithms followed, most of which were published in the 1990s \cite{ArmenS95, ArmenS98, breslauer, CzumajGPR97,KosarajuPS94,mucha,Sweedyk99,TengY93}. It is worth noting that several of these algorithms are significantly more involved than the natural \GREEDY algorithm.

The \GREEDY algorithm for SSP was proposed by Gallant, Maier, and Storer \cite{gallant}. Tarhio and Ukkonen \cite{tarhio} and independently Turner \cite{Turner89} showed that \GREEDY gives a $\frac{1}{2}$-approximation for the maximum string compression. The string compression equals the number of characters that a superstring algorithm saves from the total length of all strings in $S$, i.e., it is the total overlap between all pairs of adjacent strings across the superstring. This result, however, does not imply a constant approximation ratio upper bound for \GREEDY, for the length metric.

Moreover, Tarhio and Ukkonen showed that the approximation ratio of \GREEDY is at least 2, by considering the input $S=\{ab^{k}, b^{k+1}, b^{k}a\}$, for which, depending on the tie-breaking choice, \GREEDY will either output the shortest superstring or a superstring of length twice the minimum, when $k\rightarrow \infty$.\footnote{For $S=\{c(ab)^{k}, (ba)^{k},(ab)^{k}c\}$, \GREEDY will merge the first with the third string, producing a superstring of length twice that of the optimal superstring $c(ab)^{k+1}c$, when $k\rightarrow \infty$ \cite{blum}. No tie-breaking is involved here.} Finally, Tarhio and Ukkonen conjectured that \GREEDY is a 2-approximation algorithm, forming the long-standing \emph{Greedy Conjecture}.
By utilizing the Overlap Rotation Lemma of \cite{breslauer} in the proof of Blum et al.~\cite{blum}, Kaplan and Shafrir \cite{kaplanshafrir} showed that \GREEDY gives a 3.5-approximation.

The \GREEDY algorithm has been commonly used in practical applications when it becomes infeasible to compute an optimal solution~\cite{Li90,myers2016history,ilie2006shortest}. Also, the good performance of \GREEDY in practice has been documented within a probabilistic framework \cite{FriezeS96, Ma09}.

In this paper, we make the first progress on the approximation guarantee of \GREEDY since 2005.
\begin{theorem}\label{thm:greedy}
  The approximation ratio of \GREEDY is at most $(13+\sqrt{57})/6 \approx 3.425$.
\end{theorem}

Furthermore, we obtain a better approximation guarantee for SSP, improving slightly upon the algorithm by Mucha~\cite{mucha}.

\begin{theorem}\label{thm:ssp}
  The Shortest Superstring problem can be approximated within a factor of $(37+\sqrt{57})/18\approx 2.475$.
\end{theorem}

Finally, our techniques also imply better approximation guarantees for \TGREEDY and \MGREEDY; see Section~\ref{sec:approx}.

\section{Definitions}\label{sec:definitions}

Here, we review useful notation and concepts from previous works~\cite{blum,breslauer,kaplanshafrir} that are necessary to explain our contribution in more detail in Section~\ref{sec:ourContribution}.

By $S=\{s_{1},\dots,s_{m}\}$ we denote the input consisting of $m\ge 2$ finite strings. Without loss of generality (w.l.o.g.), we assume that no string in $S$ is a substring of another string in $S$. This is because the addition of any substring of a string in $S$ to the input cannot modify the superstring that any algorithm considered here outputs.

By $|s|$ we denote the length (i.e., number of characters) of a string $s$.
By $s[i,j]$ we denote the substring of $s$ starting at its $i$-th character and ending at its $j$-th character, where $j\in [i,|s|]$. For any two strings $s$ and $t$, $st$ will denote their concatenation.

\paragraph{Overlaps and distances.}
By $\ov(s,t)$ we denote the longest (maximum) overlap to merge a string $s$ with a different string $t$, i.e., $\ov(s,t)=s[|s|-i+1,|s|]$, where $i$ is the  largest integer for which $s[|s|-i+1,|s|]=t[1,i]$ holds. For instance, for $s$='bababa' and $t$='ababab', we have $\ov(s,t)$='ababa'. By $\ov(s,s)$ we denote the longest self-overlap of string $s$ which has length smaller than $|s|$; for instance, $\ov(s,s)$='baba' for $s$='bababa'.

By $\pref(s,t)$ we denote the prefix of maximally merging string $s$ with string $t$, i.e., assuming that $s=uv$ and $t=vz$ for strings $u$, $v=\ov(s,t)$ and $z$, it holds that $\pref(s,t)=u$.
In the same way, we define $\pref(s,s)$ so that $s = \pref(s,s)\ov(s,s)$.
The \emph{distance} $\dist(s,t)=|\pref(s,t)|$ is the number of characters of the prefix; possibly $\dist(s,t)\neq \dist(t,s)$.

\paragraph{Distance and overlap graphs.}
The \emph{distance graph} $G_{\dist}(S)=(V,E,\dist(,))$ is a complete directed graph with self-loops, where $|V|=m$, $|E|=m^{2}$. Each node corresponds to a string in $S$ and the weight of a directed edge $(s,t)$ equals $\dist(s,t)$, the distance to merge string $s$ with the (not necessarily distinct) string $t$. Note that the edge lengths satisfy the triangle inequality $\dist(s,t)\leq \dist(s,t')+\dist(t',t)$ as one always obtains the longest overlap by directly merging $s$ to $t$.

Similarly, the \emph{overlap graph} $G_{\ov}(S)$ is a complete directed graph $(V,E,|\ov(,)|)$ with self-loops, where $|V|=m$, $|E|=m^{2}$ and the profit of each directed edge $(s,t)$ equals $|\ov(s,t)|$, i.e., the longest overlap to merge string $s$ with the (not necessarily distinct) string $t$. We will also write $\ov(s,t)$ as $\ov(e)$, where $e=(s,t)$ is a directed edge of the overlap graph.

We can identify an edge $e=(s,t)$ in $G_{\dist}$ or $G_{\ov}$ with the new string $\pref(s,t)t$ which we obtain by merging $s$ and $t$. Repeating this argument, we see that a simple directed path $s_0\rightarrow s_1 \rightarrow \dots \rightarrow s_k$  corresponds to a new string $\pref(s_0,s_1)\dots\pref(s_{k-1},s_k)s_k$ which contains all strings represented by nodes on the path as substrings in the same order. Accordingly, a superstring of $S$ simply corresponds to a directed Hamiltonian path in the graph.
If two strings $s$ and $t$ appear in adjacent positions and in this order (i.e., $s$ precedes $t$) across a superstring, we say that $s$ and $t$ are \emph{merged} in the superstring.

\paragraph{Cycle Covers.}
A \CC in a complete directed weighted graph $G$ with self-loops is a set of directed cycles such that the inner degree and the outer degree of each node of $G$ are both unit. An $x$-cycle, where $x\in[1,m]$, is a directed cycle consisting of $x$ nodes.
If $s$ and $t$ are in the same cycle of a \CC containing edge $(s,t)$, we say that $s$ and $t$ are \emph{merged} in the \CC.

By $w$ we denote the minimum length  of a \CC in $G_{\dist}(S)$, i.e., $w$ is the minimum sum of distances of edges in a \CC in $G_{\dist}(S)$.  A minimum-length \CC in $G_{\dist}(S)$ is a maximum overlap \CC in $G_{\ov}(S)$, since for any edge $(s, t)$, it holds that $|\ov(s,t)| = |s| - \dist(s,t)$. Note that we may have more than one \CC with the same length $w$; to see that, consider the input $S=\{ab^{k}, b^{k+1}, b^{k}a\}$, for which the 3-cycle consisting of strings $ab^{k}, b^{k+1}, b^{k}a$ has length $k+2$, which equals the length of the 2-cycle for strings $ab^{k}, b^{k}a$ plus the length of the 1-cycle for string $b^{k+1}$.

A maximum overlap \CC  in $G_{\ov}(S)$ is computed efficiently in the second step of the \MGREEDY algorithm of Blum et al.~\cite[Theorem 10]{blum}. In a nutshell, \MGREEDY computes an optimal \CC by sorting the edges of the overlap graph non-increasingly by their overlap lengths (breaking ties arbitrarily), and adding an edge $(s,t)$ to the cycle cover if and only if no edge $(s,t')$ or $(s',t)$ has been chosen before $(s,t)$.
Fixing some arbitrary tie-breaking, we denote the resulting maximum overlap cycle cover by $\C(S)$.
For any cycle $c$ of $\C(S)$, the last edge of $c$ added by \MGREEDY to the solution is called the \emph{cycle-closing edge}. We will frequently use the fact that the overlap length of every edge in a cycle $c$ is at least as large as the overlap length of the cycle-closing edge of $c$. The sum of overlap lengths of all cycle-closing edges of $\C(S)$ will be denoted by $o$.

By $|\ALG(S)|$ we denote the length of a superstring $\ALG(S)$ produced by an algorithm \ALG for input~$S$. We use $n=|\OPT(S)|$, where \OPT is an optimal Shortest Superstring algorithm.
Since merging the last string of a superstring with the first string of this superstring gives a \CC in the distance graph (namely, a Hamiltonian cycle), it follows that $w\leq n$.

\paragraph{Representative strings.}
By $s_{c_{0}}\rightarrow s_{c_{1}} \rightarrow \dots \rightarrow s_{c_{r-1}} \rightarrow s_{c_{0}}$ we denote the cycle $c\in \C(S)$ consisting of $r\ge 1$ strings, where the last edge $s_{c_{r-1}}\rightarrow s_{c_{0}}$ always denotes the cycle-closing edge. By $R_{c}$ we denote the string $\pref(s_{c_{0}},s_{c_{1}})\pref(s_{c_{1}},s_{c_{2}})\dots\pref(s_{c_{r-2}},s_{c_{r-1}})s_{c_{r-1}}$, i.e., the string obtained by opening the cycle-closing  edge $s_{c_{r-1}}\rightarrow s_{c_{0}}$ of cycle $c$. String $R_{c}$ will be called the \emph{representative} string of cycle $c$; note that $R_c$ contains all strings of $c$ as substrings. As $\mathcal{R}$ we denote the set of all representative strings. It follows that a superstring of the strings in $\mathcal{R}$ is, also, a superstring of the strings in $S$.

\section{Our Contribution}\label{sec:ourContribution}

Our technical result is the following upper bound on $o$, the total overlap length of cycle-closing edges, in terms of the shortest superstring length $n$ and $w$, the total length of all cycles of the minimum-length cycle cover $\C(S)$:
\begin{equation}\label{eq:main_o_c_bound}
  o \le n+\alpha\cdot w  \quad\quad \text{for\ } \alpha = \frac{1 + \sqrt{57}}{6} \approx 1.425\,.
\end{equation}
This improves upon similar bounds on $o$ in~\cite{blum,kaplanshafrir}, which we outline below.
In the following two subsections, we explain how this inequality implies Theorems~\ref{thm:greedy} and~\ref{thm:ssp}.
The remaining part of the paper is devoted to proving~\eqref{eq:main_o_c_bound}.

\subsection{Improved Approximation Guarantee of \GREEDY}

Assuming that all $|E|=m^{2}$ edges of $G_{\ov}(S)$ are ordered by non-increasing overlap, breaking ties arbitrarily, \GREEDY works by going down this list and picking edge $e$ if:
\begin{itemize}
  \item $e$ does not share a head or tail with an edge $e'$ that \GREEDY picked in a previous step (such $e'$ precedes $e$ in the ordered list of edges) and
  \item $e$ is not a cycle-closing edge.
\end{itemize}
Otherwise, \GREEDY moves to the next edge in the order. Clearly, \GREEDY outputs a directed path of $m-1$ edges which gives a superstring by merging adjacent strings.
Note that the computation of $\C(S)$ by \MGREEDY only differs from \GREEDY by not using the second condition.

Blum et al.\ \cite{blum} call the edges rejected by \GREEDY for not satisfying the second condition
(but satisfying the first condition) \emph{bad back edges}.
The reason that they are called ``back edges'' is that one can number the input strings $S=\{s_{1},\dots,s_{m}\}$ so that
the superstring $\GREEDY(S)$ contains the strings in the same order, i.e.,
$s_i$ appears before $s_j$ in $\GREEDY(S)$ if and only if $i<j$.
In this subsection, we assume that the input strings are numbered in this way.

We say that a bad back edge $e$ \emph{spans interval $[i,j]$} (for $i<j$) if $e = (s_j, s_i)$.
Blum et al.\ show that the intervals spanned by two bad back edges are either disjoint or one is contained in the other,
i.e., these intervals form a laminar family \cite[Lemma~13]{blum}. A \emph{culprit} is a bad back edge $e$ such that the interval spanned by $e$ is minimal in this laminar family (i.e., there is no bad back edge $e'$ such that the interval spanned by $e'$ is properly contained in the interval spanned by $e$). See Figure~\ref{fig:culprit} for an illustration.
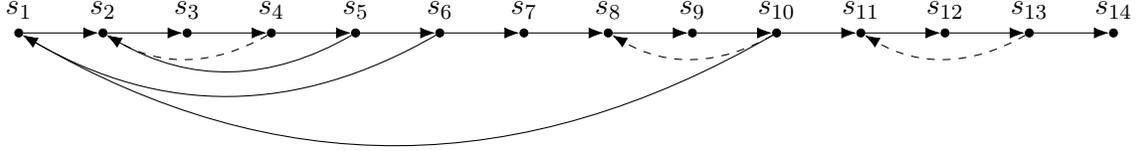
\begin{figure}
  \begin{center}
    \begin{tikzpicture}[-{Latex[length=2mm]}, vertex/.style={draw,circle,fill=black,inner sep=0pt,minimum size=3pt}]

      \node[vertex] (v1) at (0,0) [label={above:$s_1$}] {};
      \foreach \i in {2,...,14} {
          \pgfmathtruncatemacro{\prev}{\i - 1}
          \node[vertex,right=of v\prev] (v\i) [label={above:$s_{\i}$}] {};
          \draw (v\prev) to (v\i);
        }
      \foreach \i/\j in {5/2, 6/1, 10/1} {
          \draw (v\i) [bend left] to (v\j);
        }
      \foreach \i/\j in {4/2, 10/8, 13/11} {
          \draw[dashed] (v\i) [bend left] to (v\j);
        }
    \end{tikzpicture}
    \caption{Illustration of \emph{culprits}. The superstring returned by \GREEDY merges the strings $s_1$ to $s_{14}$ in this order as indicated by the path (however, the order in which \GREEDY picks edges $(s_{i-1}, s_i)$ is different). The bend edges are the \emph{bad back edges}. Out of the bad back edges, the dashed edges are the \emph{culprits}. }
    \label{fig:culprit}
  \end{center}
\end{figure}
A cycle is called \emph{culprit} if its cycle-closing edge is a culprit.

Let $w_c$ denote the sum of the lengths of culprit cycles and let $o_c$ be the sum of overlap lengths of culprit edges.
Blum et al.\ showed the following two inequalities (paragraph after the proof of Lemma 17 in \cite{blum}):
\begin{align}
  |\GREEDY(S)| & \leq 2n+o_{c}-w_{c} \label{eq:BlumFirst}
  \\
  o_{c}        & \leq n + 2w_{c} \label{eq:BlumSecond}
\end{align}
Plugging \eqref{eq:BlumSecond} into \eqref{eq:BlumFirst}, we have $|\GREEDY(S)|\leq 2n+o_c -w_c \leq 3n+w_{c} \leq 4n$, since $w_c \leq w\le n$. By using the Overlap Rotation Lemma of \cite{breslauer}, Kaplan and Shafrir \cite{kaplanshafrir} improved \eqref{eq:BlumSecond} to $o_c \leq n+1.5w_c $ and, hence, the upper bound on the approximation ratio of \GREEDY to 3.5 since $|\GREEDY(S)|\leq 2n+o_c -w_c \leq 3n+0.5\cdot w_c\leq 3.5n$.

Let $S_c\subseteq S$ be the set of input strings which lie on culprit cycles.
Blum et al.\ show that the application of \MGREEDY on $S_c$ outputs exactly the culprit cycles \cite[Lemma~15]{blum} (see also Observation~\ref{obs:CCofSubset}).
Therefore, our technical result in~\eqref{eq:main_o_c_bound} applied to input $S_c$ implies $o_{c}\leq n_{c}+\alpha\cdot w_{c}$ where $n_{c}\leq n$ equals the length of the shortest superstring for $S_c$.
Plugging this into \eqref{eq:BlumFirst}, we have:
\begin{equation}\label{eq:greedy2}
  |\GREEDY(S)|\leq 2n+n_c+(\alpha-1)\cdot w_c  \le 3n+(\alpha-1)\cdot w_c \le (2+\alpha)\cdot n \approx 3.425n\,.
\end{equation}

\subsection{Improved Approximation Guarantee for SSP}
\label{sec:approx}

As discussed before, the algorithm \MGREEDY computes $\C(S)$ or, more specifically, the set of representative strings $\mathcal{R}$ for all cycles. It then outputs the superstring that is obtained by concatenating all representative strings in an arbitrary order. The total length of the representative strings is $w+o$, i.e., the minimum length of a cycle cover in $G_{\dist}(S)$ plus the sum of overlaps of all cycle-closing edges of the cycle cover. Our main result in \eqref{eq:main_o_c_bound} states that $o\le n+\alpha\cdot w$. Therefore, the superstring computed by \MGREEDY has length $w+o\le n+(1+\alpha)\cdot w \le (2+\alpha)\cdot n$. Hence, just as for \GREEDY, we get that \MGREEDY is a $(2+\alpha)$-approximation algorithm, which improves upon the upper bound of $3.5$ implied in \cite{kaplanshafrir}.

Instead of just concatenating the representative strings, we can also attempt to overlap them, i.e., to compute a shorter superstring of the representative strings. One possibility is to use an approximation algorithm for Maximum Asymmetric TSP (MaxATSP) for this in order to find a superstring that aims to maximize the total overlap between the representative strings.

The following theorem is adopted from the literature~\cite{breslauer,Mucha07,mucha} (for this particular version we are following \cite[Theorem 21]{Mucha07}) and, combined with our new result for \MGREEDY, results in an improved approximation guarantee for SSP. A proof is included in the appendix for completeness.

\begin{theorem}\label{thm:mgreedy-to-approx}
  If \MGREEDY is a $(2+\alpha)$-approximation algorithm and there exists a $\delta$-approximation algorithm for MaxATSP (for $\delta \le 1$), then there exists a $(2+(1-\delta)\cdot \alpha)$-approximation algorithm for SSP.
\end{theorem}

Using the $\frac{2}{3}$-approximation algorithm for MaxATSP of \cite{kaplanATSP} or the more recent and simpler $\frac{2}{3}$-approximation algorithm of \cite{PaluchEZ12}, Theorem~\ref{thm:mgreedy-to-approx} with $\delta=\frac{2}{3}$ implies that we get an approximation guarantee of $\frac{37+\sqrt{57}}{18}\approx 2.475$. This improves slightly upon the approximation guarantee of $2\frac{11}{23}\approx 2.478$ of the currently best SSP algorithm~\cite{mucha}. The use of a better than $\frac{2}{3}$-approximation algorithm for MaxATSP as a black-box will give an even smaller approximation guarantee for SSP\footnote{
  Recently, Paluch~\cite{paluch} announced a $0.7$-approximation algorithm for MaxATSP, which would
  give a $2\frac{33}{76}\approx 2.434$-approximation for SSP when using the result from~\cite{breslauer,Mucha07} in a black-box way.
  Setting $\delta = 0.7$ in Theorem~\ref{thm:mgreedy-to-approx} directly implies an improved $2.428$-approximation for SSP.}.

\paragraph{\TGREEDY.}
The \TGREEDY algorithm of Blum et al.\ works by first computing the representative strings $\mathcal{R}$ and then, rather than applying a possibly complicated approximation algorithm for MaxATSP, applying \GREEDY to this set of representative strings. As \GREEDY gives a $\frac12$-approximation for such instances of MaxATSP~\cite{tarhio} (more precisely, for the longest Hamiltonian path, which is sufficient),
using $\delta=\frac{1}{2}$ in Theorem~\ref{thm:mgreedy-to-approx}, we get that \TGREEDY is a $\frac{25+\sqrt{57}}{12}\approx 2.7125$-approximation algorithm, which improves upon the upper bound of $2.75$ (implied in  \cite{breslauer,Mucha07}).

\section{The Big Picture}\label{sec:bigPicture}

\paragraph{Small, large, and extra large cycles.}
Our key idea is to partition cycles into a few types according to the ratio between their length and the overlap length of their cycle-closing edge, and treat these types differently in the analysis.
To this end, let $w(c)$ denote the length of a cycle $c$ of $\C(S)$, and let $o(c)$ denote the overlap length of the cycle-closing edge of $c$, i.e., $o(c) = |\ov(s_{c_{r-1}}, s_{c_{0}})|$, where $(s_{c_{r-1}}, s_{c_{0}})$ is the cycle-closing edge. A cycle $c$ of $\C(S)$ is
\begin{itemize}
  \item a \emph{small} cycle if $o(c)> 2w(c)$,
  \item a \emph{large} cycle if $\alpha\cdot w(c)<o(c)\le 2w(c)$, and
  \item an \emph{extra large} cycle if $o(c)\le \alpha\cdot w(c)$,
\end{itemize}
where $\alpha$ is the parameter defined in~\eqref{eq:main_o_c_bound}.
The set of extra large cycles of $\C(S)$ will be denoted by $\mathcal{X}(S)$, the set of large cycles of $\C(S)$ will be denoted by $\mathcal{L}(S)$, and the set of small cycles of $\C(S)$ will be denoted by $\mathcal{S}(S)$.

In Section \ref{sec:large43}, we show that we can assume w.l.o.g.\ that $\C(S)$ contains no extra large cycle.  For this, we exploit the slack in the right-hand side of $o(c)\le \alpha\cdot w(c)$ for an extra large cycle $c$, compared to the right-hand side of $o\leq n+\alpha\cdot w$ that we want to show.

\paragraph{Outline.}
To get our technical result in~\eqref{eq:main_o_c_bound}, we prove two independent upper bounds on $o$.
In Section~\ref{sec:firstBound}, we improve $o\le n+1.5w= n+  1.5\cdot\sum_{c\in \mathcal{S}(S)}w(c)+ 1.5\cdot\sum_{c\in \mathcal{L}(S)}w(c)$ of \cite{kaplanshafrir} to
\begin{equation}\label{eq:firstBound}
  o\leq n+  \sum_{c\in \mathcal{S}(S)}w(c)+\frac32\cdot\sum_{c\in \mathcal{L}(S)}w(c)\,.
\end{equation}

On its own, the improvement by $\frac{1}{2}\cdot\sum_{c\in \mathcal{S}(S)}w(c)$ over \cite{kaplanshafrir} is insignificant because the total length of the small cycles may be very small compared to the total length of the large cycles.
However, we show a different upper bound which is better when small cycles contribute only very little to $w$.
Namely, in Section \ref{sec:secondBound}, we prove that
\begin{equation}\label{eq:secondBound}
  o\leq n +\gamma\cdot \sum_{c\in \mathcal{S}(S)}w(c) +\sum_{c\in \mathcal{L}(S)}w(c)
\end{equation}
for a positive constant $\gamma$, and this is sufficient to obtain $o\leq n+(1.5-\epsilon)\cdot w$ for a positive constant $\epsilon$, when combined with the first upper bound on $o$. Naturally, the smaller $\gamma$ we get, the smaller the resulting upper bound.
We will require that $\gamma$ and the aforementioned parameter $\alpha$ satisfy the following four constraints:
\begin{align}
  (3-2\alpha)\cdot \gamma                           & = 2-\alpha \label{eqn:alpha_gamma_constr_extra_large}
  \\
  3\cdot \left(\alpha - \frac{2}{\gamma - 2}\right) & \ge 1 \label{eqn:alpha_gamma_constr_related_cycles}
  \\
  \frac{5}{2} + \frac{1}{2(\alpha - 1)}             & \le \gamma \label{eqn:alpha_gamma_constr_good_edge}
  \\
  \gamma                                            & \le (\gamma-1)\cdot \alpha \label{eqn:alpha_gamma_constr_sym_diff_4}
\end{align}
Solving this system of inequalities, while minimizing $\alpha$, yields
\begin{equation*}
  \alpha = \frac{1 + \sqrt{57}}{6} \approx 1.425
  \quad\quad\text{and}\quad\quad
  \gamma = \frac{31 + 3\sqrt{57}}{14} \approx 3.832\,.
\end{equation*}
Note that~\eqref{eqn:alpha_gamma_constr_good_edge} and~\eqref{eqn:alpha_gamma_constr_sym_diff_4} are not tight,
i.e., $\alpha$ and $\gamma$ are determined
by~\eqref{eqn:alpha_gamma_constr_extra_large} and~\eqref{eqn:alpha_gamma_constr_related_cycles}.

Multiplying \eqref{eq:firstBound} by $(2\alpha-2)$ and \eqref{eq:secondBound} by $(3-2\alpha)$ and adding the two resulting inequalities we get
\begin{align*}
  o & \le n + ((2\alpha-2)+(3-2\alpha)\cdot \gamma) \sum_{c\in \mathcal{S}(S)}w(c)
  \,\,+\,\, ((3\alpha-3)+(3-2\alpha))\sum_{c\in \mathcal{L}(S)}w(c)                                                   \\
    & = n +  \alpha \sum_{c\in \mathcal{S}(S)}w(c)\,\,+\,\, \alpha\sum_{c\in \mathcal{L}(S)}w(c) = n+\alpha\cdot w\,,
\end{align*}
where we use~\eqref{eqn:alpha_gamma_constr_extra_large} in the second step. This shows~\eqref{eq:main_o_c_bound}, as desired.

\paragraph{Intuition.}
\label{sec:intuition}

Before we start with formal proofs, we give some intuition and explain the main ideas behind our technical contribution.
First, we observe in Section~\ref{sec:large43} that we can assume that there are no extra large cycles (as they can be handled separately), which will come in handy for the second bound.
Note that if all (remaining) cycles are large, then our proof is complete as summing over all cycles gives $ o\leq 2\cdot w \le n + w$. On the other hand, if all cycles are small,
the first bound~\eqref{eq:firstBound} gives $o\leq n + w$, again implying a better bound than in~\eqref{eq:main_o_c_bound}. This means that it is the presence of both small and large cycles that makes the analysis challenging.

To facilitate the analysis of small cycles,
we show in Section~\ref{sec:smallStrings} that we can make the following assumption: If an optimal superstring merges two strings from one small cycle $c$, then these two strings must be merged in the small cycle $c$ as well. This essentially follows from the large amount of overlap length (relatively to $w(c)$) in small cycles.

We obtain the first bound by proving a lower bound on $n$, the optimal superstring length. Roughly speaking, we show that each small cycle $c$ must contribute at least $o(c) - w(c)$ to $n$, for which we use that strings of small cycles must be relatively long (longer than $o(c) > 2w(c)$) together with a bound from~\cite{blum} on the overlap between two strings from different cycles. For a large cycle, we use a generalization of the Overlap Rotation Lemma from~\cite{breslauer} to carefully pick a single string from this cycle that is suitable for obtaining the lower bound on $n$.

It is the second upper bound that constitutes our main technical contribution.
Recall that $w$, the length of the optimal cycle cover $\C$, is a lower bound on the length of the shortest Hamiltonian cycle $\C_0$ in $G_{\dist}$, which is itself a lower bound on $n$.
In proving the second upper bound, we make use of the difference between $w$ and the length of $\C_0$ and therefore, we derive a stronger lower bound on $n$.
Namely, we construct a careful sequence of edge swaps transforming $\C_0$ into $\C$ such that each step decreases the length of the current cycle cover by at least a certain suitable amount. In a nutshell, when an edge swap in the constructed sequence adds an edge of a small cycle $c$ to the current cycle cover, we show that this must decrease the length of the cycle cover by at least $o(c) - \gamma\cdot w(c)$ minus a term for certain large cycles affected by the swap.
Summing up over all steps will give us the desired lower bound on the length of $\C_0$.

\paragraph{Outline.} Before proving the two bounds using the ideas outlined above, we review useful lemmas from previous work in Section~\ref{sec:morePrelims}
and derive several properties of strings belonging to small cycles in Section~\ref{sec:smallStrings}.
We remark that Sections~\ref{sec:firstBound} and~\ref{sec:secondBound} are independent of each other and can be read in any order.

\section{Preliminaries for the Analysis}\label{sec:prelimsForAnalysis}

We start by observing that \MGREEDY executed on the strings belonging to a subset of cycles of the minimum cycle cover $\C(S)$ produces exactly the same subset of cycles.

\begin{observation}\label{obs:CCofSubset}
  Let $\overline{\C}\subseteq \C(S)$ be a set of cycles and let $\overline{S}\subseteq S$ be the set of input strings that belong to cycles in $\overline{\C}$.
  Then \MGREEDY on input $\overline{S}$ (with the same tie-breaking rule) outputs $\overline{\C}$, which is thus the minimum-length cycle cover of $\overline{S}$, i.e., $\C(\overline{S}) = \overline{\C}$.
\end{observation}

\begin{proof}
  Note that \MGREEDY on input $S$ rejects any edge $(s,t)$ between $\overline{S}$ and  $S\setminus \overline{S}$ because there is an incident edge $(s',t)$ or $(s,t')$ with larger (or equal) overlap that precedes $(s,t)$ in the list of edges sorted by their overlap length.
  Thus, when we run \MGREEDY on input $\overline{S}$, it selects exactly the same edges among vertices in $\overline{S}$ as when we run \MGREEDY on input $S$.
\end{proof}

\subsection{Dealing with Extra Large Cycles}\label{sec:large43}

Let $\overline{S}\subseteq S$ be the subset of strings that belong to all small and large cycles of $\C(S)$.
Observation~\ref{obs:CCofSubset} implies that $\C(\overline{S})$ consists of all small and large cycles of $\C(S)$, while
$\C(S-\overline{S})$ consists of all extra large cycles of $\C(S)$.
Let $\hat{w}$ denote the sum of lengths of the (extra large) cycles in $\C(S-\overline{S})$ and let $\hat{o}$ be the sum of overlap lengths of the cycle-closing edges of the cycles in $\C(S-\overline{S})$. Similarly, let $\overline{o}$ be the sum of overlap lengths of the cycle-closing edges in $\C(\overline{S})$ and let $\overline{w}$ be the sum of lengths of the cycles in $\C(\overline{S})$. Proving~\eqref{eq:main_o_c_bound} for input $\overline{S}$ implies that $\overline{o}\leq |\OPT(\overline{S})|+\alpha\cdot \overline{w}$, and assuming this, we show $o\leq n+\alpha\cdot w$. Indeed, we take the sum of inequality $\overline{o}\leq |\OPT(\overline{S})|+\alpha\cdot \overline{w}$ with inequality $\hat{o}\leq \alpha\cdot\hat{w}$ (which holds by the definition of extra large cycles) and obtain:
\begin{equation*}
  o = \overline{o}+\hat{o}\leq |\OPT(\overline{S})|+\alpha\cdot\overline{w}+\alpha\cdot\hat{w}=|\OPT(\overline{S})|+\alpha\cdot w\le n+\alpha\cdot w
\end{equation*}
where the penultimate step uses $w=\overline{w}+\hat{w}$ and the last inequality uses $|\OPT(\overline{S})|\le  |\OPT(S)|=n$, which follows from $\overline{S}\subseteq S$. Therefore, for proving~\eqref{eq:main_o_c_bound}, we assume w.l.o.g.~that $\C(S)$ has no extra large cycle.

\subsection{Useful Lemmas from Previous Work}\label{sec:morePrelims}

We start with describing further concepts from the literature.
A \emph{semi-infinite} string is defined as the concatenation of an infinite number of finite non-empty strings. If these strings are the same string $x$, then the semi-infinite string will be denoted by $x^{\infty}$ and called \emph{periodic}.
For a semi-infinite string $\alpha$ and integer $k\ge 1$, we denote by $\alpha[k]$ its (semi-infinite) substring which starts at its $k$-th character.

We say that a string $s$ has \emph{periodicity} of length $a$ for $a\le |s|$ if $s$ is a prefix of $x^{\infty}$ for some string $x$ of length $a$.
Note that $\pref(s,s)$ is the shortest string $x$ such that $s$ is a prefix of $x^\infty$. The length of $\pref(s,s)$ is denoted as $\period(s)=|\pref(s,s)|=\dist(s,s)$. In other words, $\period(s)$ is the smallest periodicity of a string.
We will need the following property of periodicity\footnote{%
  Blum et al.~\cite{blum} say that an \emph{equivalence class} $[s]$ (for the equivalence defined below Lemma~\ref{lm:periodicity_gcd}) has periodicity $a$
  if it is invariant under a rotation by $a$ characters, i.e., it holds that $\pref(s,s)=uv=vu$ where $|u|=a$.
  Note that this definition of periodicity is different to ours.
  With a reference to~\cite{fine1965uniqueness}, they remark that if $[s]$ has periodicities $a$ and $b$, then it has periodicity $\gcd(a, b)$ as well.
  We are not aware of a proof of this property for our definition of periodicity.}.

\begin{lemma}\label{lm:periodicity_gcd}
  Any string $s$ with periodicities $a$ and $b$ such that $|s| \ge a + b$ has periodicity $\gcd(a, b)$,
  where $\gcd(a, b)$ is the greatest common divisor of $a$ and $b$.
  Consequently, any periodicity $a$ of $s$ with $a \le |s|/2$ (if any) is an integer multiple of $\period(s)$.
\end{lemma}

\begin{proof}
  Let $g = \gcd(a, b)$, and suppose w.l.o.g.\ that $a < b$.
  Further, let $a' = a/g$ and $b' = b/g$.
  Due to the periodicity by $a$, it is sufficient to prove that for any $i = 1,\dots,a'-1$, it holds that
  $s[1,g] = s[i\cdot g + 1, (i+1)\cdot g]$.
  To this end, using that $s$ has periodicities $a$ and $b$ and that $|s| \ge a + b$, for any $i = 0,\dots,a'-1$, we get
  \begin{equation}\label{eqn:periodicity_gcd}
    s[i\cdot g + 1, (i+1)\cdot g] = s[i\cdot g + 1 + b, (i+1)\cdot g + b] = s[f(i)\cdot g + 1, (f(i)+1)\cdot g]\,,
  \end{equation}
  where we use $f(i) := (i + b')\bmod a'$.
  Note that the cyclic group $\mathbb{Z}/a'\mathbb{Z} = \{0,\dots,a'-1\}$ of integers modulo $a'$ (with addition) is generated by $b'\bmod a'$,
  since $\gcd(a', b') = \gcd(a/g, b/g) = 1$, which follows from $\gcd(a, b) = g$.
  Hence, applying~\eqref{eqn:periodicity_gcd} for $i = 0, f(0), f(f(0)),$ and so on proves that
  $s[1,g] = s[i\cdot g + 1, (i+1)\cdot g]$ for any $i = 1,\dots,a'-1$.
\end{proof}

String $z$ is a \emph{rotation} of string $q$ if $q=uv$ and $z=vu$ for some strings $v$ and $u$
(string $z$ is a rotation of itself if one of them is empty).
Two strings $s$ and $t$ are \emph{equivalent} if $\pref(t,t)$ is a rotation of $\pref(s,s)$, i.e., there exist strings $x$ and $y$ (possibly empty) such that $\pref(s,s)=xy$ and $\pref(t,t)=yx$. Two strings that are not equivalent will be called \emph{inequivalent}.

For any cycle $c=s_{c_{0}}\rightarrow s_{c_{1}} \rightarrow \dots \rightarrow s_{c_{r-1}} \rightarrow s_{c_{0}}$ in $G_{\dist}(S)$, we define $s(c)$ as the string $\pref(s_{c_{0}},s_{c_{1}})\pref(s_{c_{1}},s_{c_{2}})\dots\pref(s_{c_{r-1}},s_{c_{0}})$. Note that $R_c = s(c)\ov(s_{c_{r-1}},s_{c_{0}})$ and $R_c$ is a prefix of $s(c)^\infty$. We define as $\strings(c,s_{c_{l}})$ the string $\pref(s_{c_{l}},s_{c_{l+1}})\dots\pref(s_{c_{l-1}},s_{c_{l}})$, where  subscript arithmetic is modulo $r$  and $0\le l\le r-1$. In other words, $\strings(c,s_{c_{l}})$ is a rotation of $s(c)$ such that string $s_{c_{l}}$ is a prefix of $\strings(c,s_{c_{l}})^\infty$.

The following three lemmas appear in previous works:

\begin{lemma}[Claim 2 in \cite{blum}]\label{lm:Claim2Blum}
  For any cycle $c$ in the distance graph for $S$, every string of $c$ is a substring of $s(c)^{\infty}$.
\end{lemma}

\begin{lemma}[Claim 3 in  \cite{blum}]\label{lm:Claim3Blum}
  If all strings of a subset of $S$ are substrings of a semi-infinite string $t^{\infty}$, then there exists a cycle of length $|t|$ in the distance graph $G_{\dist}(S)$ that contains all these strings.
\end{lemma}

\begin{lemma}[Lemma 13 in \cite{Mucha07}] \label{lm:lmMucha}
  It holds that $\period(R_{c})=w(c)$ for any cycle $c$ of $\C(S)$.
\end{lemma}

As a corollary of these lemmas, we obtain:

\begin{observation}\label{obs:representativesInequivalent}
  The representative strings $R_c$ and $R_{c'}$ for any two cycles $c$ and $c'$ in $\C(S)$ are inequivalent.
  Moreover, any string $\hat{R}_{c'}$ that contains all strings of cycle $c'$ as substrings is inequivalent to $s(c)^\infty$.
\end{observation}

\begin{proof}
  Recall that $R_c$ is a prefix of $s(c)^\infty$, which contains all strings of $c$ by Lemma~\ref{lm:Claim2Blum}.
  From Lemma~\ref{lm:lmMucha} it follows that $\pref(R_c,R_c) = s(c)$ and $\pref(R_{c'},R_{c'}) = s(c')$.
  If  $R_c$ and $R_{c'}$ were equivalent, then $s(c')$ is a rotation of $s(c)$ and thus, any string of both cycles appears as a substring of $s(c)^\infty$.
  Therefore, by Lemma~\ref{lm:Claim3Blum}, all strings of both $c$ and $c'$ are contained in a single cycle of length $w(c)$, contradicting the minimality of $\C(S)$.

  The second claim follows similarly. If $\pref(\hat{R}_{c'},\hat{R}_{c'})$ is a rotation of $f := \pref(s(c)^\infty,s(c)^\infty)$,
  then $f^\infty = s(c)^\infty$ contains all strings of both $c$ and $c'$,
  so we again obtain a contradiction with the minimality of $\C(S)$ by using Lemma~\ref{lm:Claim3Blum}.
\end{proof}

Since the representative string $R_{c}$ contains any string $s$ of the cycle $c$ it belongs to, the period of $s$ cannot be larger than $\period(R_{c})$ and thus, by Lemma~\ref{lm:lmMucha}, we obtain:

\begin{observation}\label{obs:periodUB}
  For any string $s$ of a cycle $c\in \C(S)$, it holds that $\period(s)\le w(c)$.
\end{observation}

Next, we  need the following upper bound for the overlap length between inequivalent strings:

\begin{lemma}[Lemma 2.3 in \cite{kaplanshafrir}]\label{lm:kaplanShafrirOverlap}
  For any two inequivalent strings $s$ and $t$, it holds that $|\ov(s,t)|< \period(s)+\period(t)$.
\end{lemma}

In the case that these two inequivalent strings belong to two different cycles $c$ and $c'$ of $\C(S)$, we have
$|\ov(s,t)| <w(c)+w(c')$ by Observation~\ref{obs:periodUB}, and more generally:

\begin{lemma}[Lemma 9 in \cite{blum}]\label{lm:maOverlap}
  Let $c$ and $c'$ be any two cycles of $\C(S)$. It holds that $|\ov(s,t)| <w(c)+w(c')$, where $s$ is any string of $c$ and $t$ is any string of $c'$.
\end{lemma}

We will need an even more general corollary that follows from the same argument as in Lemma~9 in \cite{blum}
(see also Lemma~7 in \cite{Mucha07}), but we provide a proof for completeness.

\begin{corollary}\label{cor:lm:maOverlap}
  Let $c$ and $c'$ be any two cycles of $\C(S)$. Any string $h$, which is a substring of both $s(c)^\infty$ and $s(c')^\infty$,
  satisfies $|h| < w(c)+w(c')$.
  In particular, it holds that $|\ov(s,t)| <w(c)+w(c')$, where $s$ is any substring of $s(c)^\infty$ and $t$ is any substring of $s(c')^\infty$.
\end{corollary}

\begin{proof}
  Assume for a contradiction that $|h| \ge w(c)+w(c')$.
  Since $h$ is a substring of $s(c)^\infty$, it is a prefix of $x_1^\infty$ for a string $x_1$ with $|x_1| = w(c)$, which is a rotation  of $s(c)$.
  Similarly, $h$ is a prefix of $x_2^\infty$ for $x_2$ with $|x_2| = w(c')$, which is a rotation of $s(c')$.
  Using $|h| \ge w(c)+w(c')$, we get that $x_1 x_2 = x_2 x_1$ and by a simple induction, it holds that $x_1^k x_2^k = x_2^k x_1^k$
  for any $k\ge 1$, which implies $x_1^\infty = x_2^\infty$. Since any string in cycle $c$ is a substring of $s(c)^\infty$, it is also a substring of $x_1^\infty = x_2^\infty$, and similarly for $c'$. Thus, using Lemma~\ref{lm:Claim3Blum} gives a contradiction with the fact that $c$ and $c'$ are two cycles of the minimum-length \CC $\C(S)$.
\end{proof}

\subsection{Properties of Strings of Small Cycles}\label{sec:smallStrings}

In this section, we prove several properties of small cycles.
Consider a small cycle  $c$.
Recall that the \MGREEDY algorithm picks edges in non-increasing order of overlap length when producing $\C(S)$. Therefore, $o(c)$ is no larger than any other overlap length between two merged strings in cycle $c$. By this and since the length of any string $s$ in $c$
is greater than the length of any of its two (i.e., left and right) overlaps  (or the self-overlap if $c$ is a 1-cycle),
we have $|s| > o(c)$.
Further, by the definition of a small cycle, it is $o(c)>2\cdot w(c)$ and thus, for any string $s$ of $c$, we get:
\begin{equation}\label{eq:minStringLength0}
  |s|> 2\cdot w(c)
\end{equation}
Note that the representative string $R_c$ is even longer as $|R_{c}|=w(c)+o(c)> 3\cdot w(c)$, since string $R_{c}$ is formed by opening cycle $c$ at the cycle-closing edge.

While a string of a cycle $c$ is not necessarily equivalent to string $R_c$ (cf.\ Lemma~2.1 in~\cite{kaplanshafrir}),
we prove that this property actually holds for small cycles.

\begin{lemma}\label{lm:equivSmall}
  Consider any small cycle $c$ of $\C(S)$. All strings of $c$ and $R_{c}$ are equivalent and in particular,
  $\period(s) = w(c)$ for any string $s$ of cycle $c$.
\end{lemma}

\begin{proof}
  Recall that $R_c$ is a prefix of $s(c)^\infty$.
  From Lemma~\ref{lm:lmMucha} it follows that $\pref(R_c,R_c) = s(c)$.
  Hence, it suffices to show that $\pref(s,s)$ is a rotation of $s(c)$ for any string $s$ of the small cycle $c$.
  We first prove that $\period(s) = w(c)$.
  By Observation~\ref{obs:periodUB}, we have $\period(s)\le w(c)$. Assume for a contradiction that $\period(s)<w(c)$.
  Since $s$ has periodicity $w(c)$ and, by~\eqref{eq:minStringLength0}, $|s|> 2w(c)$, we have that $w(c)$ must be a multiple of $\period(s)$ by Lemma~\ref{lm:periodicity_gcd}.
  So there exists an integer $k\ge2$ such that $k\cdot |\pref(s,s)|=k\cdot \period(s)=w(c)$.
  Recall that $\strings(c,s)$ is a rotation of $s(c)$ that is a prefix of $s$ and has length $w(c)$.
  We thus have that $\strings(c,s) = \pref(s,s)^k$, which implies $\strings(c,s)^\infty = \pref(s,s)^\infty$.
  Note that every substring of $s(c)^\infty$ is also a substring of $\strings(c,s)^\infty = \pref(s,s)^\infty$. By Lemmas~\ref{lm:Claim2Blum} and~\ref{lm:Claim3Blum}, it follows that all strings of $c$ belong to a cycle (in $G_{\dist}(S)$) of length $|\pref(s,s)|=\period(s)<w(c)$, which contradicts the minimality of $\C(S)$.
  Hence, $\period(s) = w(c)$ and thus, $\pref(s,s) = \strings(c,s)$. This concludes the proof as $\strings(c,s)$ is a rotation of $s(c) = \pref(R_c,R_c)$.
\end{proof}

As a corollary, we obtain that for small cycles, the triangle inequality in $G_{\dist}(S)$
becomes equality.

\begin{lemma}\label{lm:smallStrings}
  Consider two strings $s\in S$ and $t\in S$ both belonging to a small cycle $c\in \C(S)$ and assume that $s$ is not merged with $t$ across cycle $c$. Then, for any string $t'$ that lies on cycle $c$ between $s$ and $t$ (in this order), it holds that $\dist(s,t)=\dist(s,t')+\dist(t',t)$.
\end{lemma}

\begin{proof}
  First, it is $\dist(s,t)\leq \dist(s,t')+\dist(t',t)$ by the triangle inequality in $G_{\dist}(S)$. Next, assume for a contradiction that $\dist(s,t)<\dist(s,t')+\dist(t',t)$.
  Consider the semi-infinite string $R' = \pref(s, t)\strings(c, t)^\infty$. Let $t_0 = t, t_1, \dots, t_\ell = s$ be the strings on the directed path from $t$ to $s$ on cycle $c$. Observe that $s$ is a prefix of $R'$ (as $t$ is a prefix of $\strings(c, t)^\infty$) and a substring of $\strings(c, t)^\infty$, starting at position $\sum_{j=0}^{\ell-1} \dist(t_j, t_{j+1})$. It follows that
  $$\dist(s,s) \le \dist(s, t) + \sum_{j=0}^{\ell-1} \dist(t_j, t_{j+1}) < \dist(s,t')+\dist(t',t) + \sum_{j=0}^{\ell-1} \dist(t_j, t_{j+1}) \le w(c)\,,$$
  where the penultimate inequality holds by the assumption $\dist(s,t)<\dist(s,t')+\dist(t',t)$ and the last inequality follows by using the triangle inequality in $G_{\dist}(S)$ for the edges between $s$ and $t'$ and for those between $t'$ and $t$. Thus, we have that $\period(s)=\dist(s,s)<w(c)$, which contradicts Lemma \ref{lm:equivSmall}.
\end{proof}

Lemma~\ref{lm:smallStrings} implies the following useful property:

\begin{observation}\label{obs:smallCyclesRotations}
  If two strings that belong to the same small cycle $c\in \C(S)$ are \emph{not} merged in $c$, then there is an optimal superstring in which they are \emph{not} merged.
\end{observation}

\begin{proof}
  Suppose that strings $s, t$ belonging to $c\in \C(S)$ are not merged in $c$,
  and let $t_1, \dots, t_\ell$ (for $\ell\ge 1$) be the strings on the directed $s$-$t$-path in $c$.
  Let $\sigma$ be any superstring in which $s$ and $t$ are merged.
  Consider string $\hat{\sigma}$ obtained by removing strings $t_1, \dots, t_\ell$ from $\sigma$, which may only decrease its length, i.e., $|\hat{\sigma}|\le |\sigma|$.
  From $\hat{\sigma}$, we create a superstring $\sigma'$ by inserting strings $t_1, \dots, t_\ell$ between $s$ and $t$ in $\hat{\sigma}$.
  Crucially, by Lemma~\ref{lm:smallStrings}, it holds that $|\sigma'| = |\hat{\sigma}|\le |\sigma|$.
  Thus, if $\sigma$ is optimal, then $\sigma'$ is also optimal.
\end{proof}

\begin{remark}\label{rem:allTogether}
  By Observation \ref{obs:smallCyclesRotations}, if a superstring $\sigma$ merges all $r$ strings belonging to the same small cycle $c=s_{c_{0}}\rightarrow s_{c_{1}} \rightarrow \dots \rightarrow s_{c_{r-1}} \rightarrow s_{c_{0}}$ (i.e., they all appear in adjacent positions across the superstring $\sigma$), then we can transform $\sigma$ into a superstring $\sigma'$ with $|\sigma'|\le |\sigma|$ where the order of these strings across $\sigma'$ is a rotation of the ordered set $\{s_{c_{0}},s_{c_{1}},\dots,s_{c_{r-1}}\}$. In this case, each of the $r$ edges of $c\in\C(S)$ coincides with an edge of $\sigma'$ except for one edge, which is not necessarily the cycle-closing edge $s_{c_{r-1}} \rightarrow s_{c_{0}}$ of $c$.
\end{remark}

\section{The First Upper Bound}\label{sec:firstBound}
In this section, we prove~\eqref{eq:firstBound}, which is our first bound on $o$.

We consider a partition of strings of all small cycles such that no two strings from two different cycles are in one part and moreover, due to Observation \ref{obs:smallCyclesRotations}, if strings $s$ and $t$ from a small cycle $c$ are in one part, then all strings between $s$ and $t$ on $c$ are in that part as well. In other words, this partition consists of directed paths and single nodes that remain after removing a subset of edges from small cycles. The particular partition that we consider below is induced by an optimal superstring for a certain subset of the input $S$ containing all strings of small cycles and one (carefully chosen) string of each large cycle.

Consider a small cycle $c$.
Let $r'$ be the number of parts with strings from cycle $c$, and for $j = 0,\dots, r'$, denote by $\bar{s}_{j}$ the string obtained by merging strings in the $j$-th part (in the same order as they appear on the small cycle $c$).
In the next technical lemma, we lower-bound the sum of lengths of the strings $\bar{s}_{j}$.

\begin{lemma}\label{lm:residuesSmall}
  It holds that $\sum_{j=0}^{r'-1}(|\bar{s}_{j}|-2\cdot w(c))\ge o(c)-w(c)$ for any small cycle $c=s_{c_{0}} \rightarrow \dots \rightarrow s_{c_{r-1}} \rightarrow s_{c_{0}}$, where $r'\leq r$.
\end{lemma}

\begin{proof}
  Fix a small cycle $c$.
  Consider string $\bar{s}_{j}$, and let $t_{j}^{0}, t_{j}^{1}, \dots, t_{j}^{\ell_j-1}$ for $\ell_j\ge 1$ be the strings that are merged into $\bar{s}_{j}$.
  Assuming that the parts are numbered in the order in which they appear on the cycle,
  $t_{j+1}^{0}$ is the string to which $t_{j}^{\ell_j-1}$ is merged on cycle $c$, with the subscript arithmetic modulo $r'$.
  (In the special case of a 1-cycle, we have $r' = r = 1$, $\ell_0 = 1$, $t^0_0$ is the only string of that cycle, and we use $t^0_1 = t^0_0$.)
  It holds that:
  \begin{align*}
    |\bar{s}_{j}|
     & =\sum_{k=0}^{\ell_j-2}\dist(t_{j}^{k},t_{j}^{k+1})+|t_{j}^{\ell_j-1}|
    \\
     & =\sum_{k=0}^{\ell_j-2}\dist(t_{j}^{k},t_{j}^{k+1})+\dist(t_{j}^{\ell_j-1},t_{j+1}^{0}) + |\ov(t_{j}^{\ell_j-1},t_{j+1}^{0})| \,,
  \end{align*}
  since $|s| = \dist(s,t) + |\ov(s,t)|$ for any two strings $s$ and $t$.
  Summing over all $r'$ strings $\bar{s}_{j}$, we get
  \begin{align*}
    \sum_{j=0}^{r'-1}(|\bar{s}_{j}|-2 w(c))
     & = \sum_{j=0}^{r'-1}\left( \sum_{k=0}^{\ell_j-2}\dist(t_{j}^{k},t_{j}^{k+1})+\dist(t_{j}^{\ell_j-1},t_{j+1}^{0}) + |\ov(t_{j}^{\ell_j-1},t_{j+1}^{0})| -2w(c) \right)
    \\
     & = w(c) + \left(|\ov(t_{0}^{\ell_0-1},t_{1}^{0})| -2w(c) \right) + \sum_{j=1}^{r'-1} \left(|\ov(t_{j}^{\ell_j-1},t_{j+1}^{0})| -2w(c) \right)
    \\
     & \ge w(c) + \left(o(c)-2w(c)\right) + 0 = o(c)-w(c)\,,
  \end{align*}
  where the second equality uses that each edge of cycle $c$ either ``lies inside a string $\bar{s}_{j}$'', i.e., is an edge $(t_{j}^{k},t_{j}^{k+1})$ for some $j$ and $0\le k\le \ell_j-2$, or ``leads from string $\bar{s}_{j}$ to $\bar{s}_{j+1}$'', i.e., is an edge $(t_{j}^{\ell_j-1},t_{j+1}^{0})$ for some $j$, and the inequality follows from the fact that $o(c)$ is the smallest overlap on cycle $c$ and that $o(c) > 2w(c)$ as the cycle is small.
\end{proof}

We will need the Overlap Rotation Lemma from~\cite{breslauer}:
\begin{lemma}[Lemma 3.3 in~\cite{breslauer}] \label{lm:overlapRotationLemma}
  Let $\alpha$ be a periodic semi-infinite string. There exists an integer $k\in [1, \period(\alpha)]$ such that $|\ov(s,\alpha[k])|< \period(s)+\frac{1}{2}\period(\alpha)$ for any (finite) string $s$ inequivalent to  $\alpha$.
\end{lemma}

Note that the index $k$ is universal for all strings inequivalent to $\alpha$.
We now generalize Lemma~\ref{lm:overlapRotationLemma}:

\begin{lemma}\label{lm:overlapRotationLemmaGen}
  Let $\alpha$ and $k$ be as in Lemma \ref{lm:overlapRotationLemma}. For any $k'\in[0,k)$ and any (finite) string $s$ inequivalent to $\alpha$, the string $\alpha[k-k']$ satisfies $|\ov(s,\alpha[k-k'])|< \period(s)+\frac{1}{2}\period(\alpha)+k'$.
\end{lemma}

\begin{proof}
  For $k'=0$ the statement of the lemma coincides with Lemma~\ref{lm:overlapRotationLemma}. It remains to show the lemma for $k'>0$.
  We have
  \begin{align*}
    |\ov(s,\alpha[k-k'])| & = |s|-\dist(s,\alpha[k-k'])
    \\
                          & \le |s| - \dist(s,\alpha[k]) + \dist(\alpha[k-k'],\alpha[k])
    \\
                          & = |\ov(s,\alpha[k])| + \dist(\alpha[k-k'],\alpha[k])
    \\
                          & \le  |\ov(s,\alpha[k])| + k' < \period(s)+\frac{1}{2}\period(\alpha) + k'\,,
  \end{align*}
  where in the second line, we applied the triangle inequality in $G_{\dist}(S)$ and the last step follows from Lemma~\ref{lm:overlapRotationLemma}.
\end{proof}

In Lemma \ref{lm:newOPTBound}, we prove the first upper bound on $o$, i.e., inequality~\eqref{eq:firstBound}.
\begin{lemma}\label{lm:newOPTBound}
  It holds that $o\leq n +  \sum_{c\in \mathcal{S}(S)}w(c) + 1.5\cdot\sum_{c\in \mathcal{L}(S)}w(c)$.
\end{lemma}

\begin{proof}
  First, for each large cycle $c$, we apply Lemma \ref{lm:overlapRotationLemma} for the semi-infinite string $\alpha_c = s(c)^\infty$ to get an integer $k_c\ge 1$. We also let $k'_c$ to be the smallest integer $k'\ge 0$ such that $\alpha_c[k_c-k']$ starts with a string $f_c$ from cycle $c$.
  By the minimality of $k'_c$, it follows that $k'_c < \dist(f_c, t_c)$, which is the prefix length between $f_c$ and string $t_c$ that $f_c$ is merged to across the large cycle $c$. See the following for an illustration.

  \[
    \alpha_c=\underbrace{abcabcabc}_{\pref(s_{c_{0}},s_{c_{1}})}\underbrace{abcabcabcabc}_{\pref(s_{c_{1}},s_{c_{2}})}\underbrace{\tikzmarknode{prefstart}{a}bcabca\tikzmarknode{k}{b}cabcabcabc}_{\pref(\tikzmarknode{fc}{f_c=s_{c_{2}}},t_c=s_{c_{3}})}\underbrace{abcabcabcabcabca}_{\pref(s_{c_{3}},s_{c_{0}})}\underbrace{abcabcabc}_{\pref(s_{c_{0}},s_{c_{1}})}\dots
  \]
  \begin{tikzpicture}[remember picture,overlay]
    \node[above=0.3cm of k] (k label) {\scriptsize $k_c$};
    \node[above=0.3cm of prefstart] (prefstart label) {};
    \draw[-] (k label) -- ([yshift=3pt] k.north);
    \path (prefstart label|-k label) edge[<->] node[fill=white] {\scriptsize $k'_c$} (k label);
  \end{tikzpicture}

  \noindent Since $|f_c|=\dist(f_c,t_c)+|\ov(f_c,t_c)|$ and $|\ov(f_c,t_c)| \ge o(c)$,  we get that $k'_c < \dist(f_c,t_c) = |f_c|-|\ov(f_c,t_c)| \le |f_c|-o(c)$.

  Fix input $S_{r}\subseteq S$, which contains all strings of $S$ belonging to small cycles and only the single string $f_c$ from each large cycle $c$. Consider $\OPT(S_{r})$, the optimal superstring of $S_{r}$, and let $n_{r}=|\OPT(S_{r})|$. Our aim is to derive a lower bound on $n_{r}\le n$.

  Superstring $\OPT(S_{r})$ induces a partition of the strings in each small cycle $c$ such that strings in each part are merged together in $\OPT(S_{r})$, while strings from different parts are separated by a string from a different cycle; this is the partition for which we apply Lemma~\ref{lm:residuesSmall}. By Observation~\ref{obs:smallCyclesRotations}, we may assume that the order in which strings of the same small cycle $c$ are merged in $\OPT(S_{r})$ is the same as the order in which they appear on $c$. For a small cycle $c$, let $r_c'$ be the size of this partition of strings in $c$, and for $j = 0,\dots, r_c'$, denote by $\bar{s}_{c,j}$ the string obtained by merging strings in the $j$-th part (in the same order as they appear on $c$).

  The key step towards lower-bounding $n_{r}$ is to obtain suitable upper bounds on the overlap length of two strings merged in $\OPT(S_{r})$ after we merge strings of small cycles $c$ to obtain strings $\bar{s}_{c,j}$.
  First, consider string $f_c$ of a large cycle $c$ and string $s'$ from a cycle $c'$ for $c'\neq c$ such that $s'$ is either $f_{c'}$ or $\bar{s}_{c',j}$ (depending on whether $c'$ is large or small) and $s'$ and $f_c$ are merged in $\OPT(S_{r})$ in this order.
  Consider string $\hat{R}_{c'} := \strings(c',s')s'$\footnote{Strictly speaking, $\strings(c',s')$ is only defined for a string $s'$ of cycle $c'$. If $c'$ is a small cycle and $s' = \bar{s}_{c',j}$ is a result of merging strings $t_{j}^{0}, t_{j}^{1}, \dots, t_{j}^{\ell_j-1}$ from cycle $c'$, then we let $\strings(c',s') := \strings(c',t_{j}^{0})$ so that $\hat{R}_{c'} = \strings(c',t_{j}^{0})s'$.}.
  Note that $\period(\hat{R}_{c'}) \le w(c')$ as $\hat{R}_{c'} = \strings(c',s')s'$, $s'$ is a prefix of $\strings(c',s')^\infty$ and $|\strings(c',s')| = w(c')$.
  Furthermore, $\hat{R}_{c'}$ contains all strings of cycle $c'$ as substrings, and thus, $\hat{R}_{c'}$ is inequivalent to $\alpha_c$ by Observation~\ref{obs:representativesInequivalent}.
  Since $f_c$ is a prefix of $\alpha_c[k_c-k'_c]$ and $s'$ is a suffix of $\hat{R}_{c'}$, we have
  $|\ov(s', f_c)| \le |\ov(\hat{R}_{c'}, \alpha_c[k_c-k'_c])|$. Using this together with Lemma~\ref{lm:overlapRotationLemmaGen} for $\alpha_c$, $k_c'$, and $\hat{R}_{c'}$, it holds that
  \begin{align}\label{eq:f_c_overlap}
    |\ov(s', f_c)| \le |\ov(\hat{R}_{c'}, \alpha_c[k_c-k'_c])|
     & < \period(\hat{R}_{c'}) + \frac12 \period(\alpha_c) + k_c'
    \nonumber                                                     \\
     & < w(c') + \frac12 w(c) + |f_c|-o(c)\,,
  \end{align}
  where the third inequality uses $\period(\hat{R}_{c'}) \le w(c')$,
  $\period(\alpha_c) \le w(c)$ (by the definition of $\alpha_c = s(c)^\infty$ and $|s(c)| = w(c)$),
  and $k_c'<|f_c|-o(c)$.

  Second, consider string $\bar{s}_{c,j}$ for a small cycle $c$ (recall that $\bar{s}_{c,j}$ may be the result of merging several strings appearing consecutively on $c$). Let $s'$ be the string merged to $\bar{s}_{c,j}$ in $\OPT(S_{r})$ in this order, and let $c'$ be the (large or small) cycle of string $s'$. From Corollary~\ref{cor:lm:maOverlap} we get
  \begin{equation}\label{eq:bar-s_c,j_overlap}
    |\ov(s', \bar{s}_{c,j})| < w(c') + w(c)\,.
  \end{equation}
  Observe that $n_{r} \ge \sum_s (|s| - |\ov(s', s)|)$, where the sum is over strings $f_c$ and $\bar{s}_{c,j}$ as defined above and
  $s'$ is the string merged to $s$ in $\OPT(S_{r})$ ($s'$ is empty for the first string in $\OPT(S_{r})$). Next, we use~\eqref{eq:f_c_overlap} or~\eqref{eq:bar-s_c,j_overlap} to bound $|\ov(s', s)|$ for all such strings $s$.
  In particular, since each such string appears once as string $s'$ (except for the last one), we get that
  \begin{equation}\label{eq:sumLengths}
    n_{r}\ge \sum_{c\in \mathcal{L}(S)} (|f_c|-1.5\cdot w(c)-(|f_c|-o(c)))+ \sum_{c\in \mathcal{S}(S)}\sum_{j=0}^{r_c'-1}\big(|\bar{s}_{c,j}| - 2\cdot w(c)\big)\,.
  \end{equation}
  Using Lemma \ref{lm:residuesSmall}, we lower-bound the second term in the right-hand side of \eqref{eq:sumLengths} and obtain
  \begin{equation*}\label{eq:sumLengthsNew}
    n_{r}\ge \sum_{c\in \mathcal{L}(S)} \big(o(c)-1.5\cdot w(c)\big)+ \sum_{c\in \mathcal{S}(S)} \big(o(c)-w(c)\big)
  \end{equation*}
  Using that $n=|\OPT(S)|\ge |\OPT(S_{r})|=n_{r}$ as $S_{r}\subseteq S$, and that $o= \sum_{c\in\mathcal{L}(S)} o(c)+\sum_{c\in\mathcal{S}(S)} o(c)$, we obtain
  \begin{equation*}
    n\ge o - 1.5\cdot \sum_{c\in \mathcal{L}(S)} w(c) - \sum_{c\in \mathcal{S}(S)} w(c)\,,
  \end{equation*}
  which completes the proof by rearranging.
\end{proof}

\section{The Second Upper Bound}\label{sec:secondBound}

In this section we show \eqref{eq:secondBound}. The first ingredient of our analysis is a suitable modification of the input set of strings $S$.

\subsection{Modifying the Input}\label{sec:modifyingInput}
For each small cycle $c=s_{c_{0}}\rightarrow s_{c_{1}} \rightarrow \dots \rightarrow s_{c_{r-1}} \rightarrow s_{c_{0}}$ in $\C(S)$, we remove all strings belonging to this cycle from $S$ and instead add the string $$R'_{c} := \pref(s_{c_{0}},s_{c_{1}})\pref(s_{c_{1}},s_{c_{2}})\dots \pref(s_{c_{r-2}},s_{c_{r-1}})\pref(s_{c_{r-1}},s_{c_{0}})s_{c_{0}}$$ to $S$. Note that the representative string $R_c$ is a prefix of $R'_c$ and thus, $R'_c$ contains all strings of the small cycle $c$. We denote the new set of strings obtained this way by $S'$.

The length of $\C(S')$ is the same as the length of $\C(S)$. Indeed, due to Lemma~\ref{lm:lmMucha}, the generated optimal \CC remains the same except that whenever we had a small cycle $c$ involving nodes $s_{c_{0}}, s_{c_{1}}, \dots, s_{c_{r-1}}$ before, we now only have a single node (corresponding to the string $R'_c$) and a self-loop at that node. In addition, the length of small cycles does not change, i.e., $\sum_{c\in \mathcal{S}(S')}w(c) = \sum_{c\in \mathcal{S}(S)}w(c)$, again by Lemma~\ref{lm:lmMucha}.

However, the length $n'=|\OPT(S')|$ of the shortest superstring of $S'$ could increase compared to the length $n=|\OPT(S)|$ of the optimal shortest superstring of $S$. The following lemma gives a bound on the increase.
\begin{lemma}\label{lem:input-transform}
  The shortest superstring for $S'$ is at most by $\sum_{c\in \mathcal{S}(S)}w(c)$ longer than the shortest superstring for $S$.
\end{lemma}
\begin{proof}
  We show how to transform any superstring $\sigma$ for $S$ into a superstring $\sigma'$ for $S'$ (which is also a superstring for $S$ as $R'_c$ contains all strings of the small cycle $c$)
  while only increasing the length of the superstring by $\sum_{c\in \mathcal{S}(S)}w(c)$, i.e., $|\sigma'| \le |\sigma| + \sum_{c\in \mathcal{S}(S)}w(c)$.
  Namely, for every small cycle  $s_{c_{0}}\rightarrow s_{c_{1}} \rightarrow \dots \rightarrow s_{c_{r-1}} \rightarrow s_{c_{0}}$ in $\C(S)$, we replace the first occurrence of $s_{c_{0}}$ in $\sigma$ by $R'_c$. The resulting superstring is our new string $\sigma'$, which
  by construction, contains all strings of $S'$ as required.

  For a small cycle $c$, the length of $R'_c$ is equal to $|s_{c_{0}}|+w(c)$. Therefore, $|\sigma'| \le |\sigma| + \sum_{c\in \mathcal{S}(S)}w(c)$ as claimed.
\end{proof}

\begin{corollary}\label{cor:opt-lower-bound}
  Let $\C_0(S')$ be a directed Hamiltonian cycle of minimum length in the distance graph $G_{\dist}(S')$. The length $n$ of the shortest superstring for $S$ is at least $|\C_0(S')| - \sum_{c\in \mathcal{S}(S')}w(c)$.
\end{corollary}
\begin{proof}
  The length $n'$ of the shortest superstring for $S'$ is at least $|\C_0(S')|$, since
  we can form a Hamiltonian cycle of length at most $n'$ by merging the first and last string of the shortest superstring.
  With this, the corollary follows from Lemma~\ref{lem:input-transform}.
\end{proof}

Since the sum of overlap lengths of cycle-closing edges in $\C(S')$, denoted $o'$, cannot be smaller than $o$, the sum of overlap lengths of cycle-closing edges in $\C(S)$, showing the following inequality
\begin{equation}\label{eq:secondBoundNew}
  o'\leq |\C_0(S')| +  (\gamma - 1)\cdot\sum_{c\in \mathcal{S}(S')}w(c)+  \sum_{c\in \mathcal{L}(S')}w(c)
\end{equation}
implies \eqref{eq:secondBound}, due to Corollary \ref{cor:opt-lower-bound}.

\subsection{Overview of the Proof}
Before proceeding, we note that our goal is to show \eqref{eq:secondBoundNew} and from now on we will only be concerned with the modified input $S'$. Therefore, for the sake of simplicity, we omit the set $S'$ from the cycle cover notation from this point onward (for instance, we shall indicate $\C(S')$ as $\C$ and $\C_{0}(S')$ as $\C_{0}$).

Consider a maximum directed Hamiltonian cycle $\C_{0}$ in $G_{\ov}(S')$ and note that $\C_{0}$ is, in particular, also a (not necessarily maximum) \CC in $G_{\ov}(S')$. We call the sum of the profits of the edges of a \CC in $G_{\ov}(S')$ the \emph{total overlap} of the \CC. Our goal is to show that the total overlap
of $\C_{0}$ is by at least
\begin{equation}\label{eqn:overlap-difference}
  \sum_{c\in \mathcal{S}(S')} (o(c)-\gamma\cdot w(c))+\sum_{c\in \mathcal{L}(S')} (o(c)-2\cdot w(c))
\end{equation}
smaller than the total overlap of the optimal \CC $\C$.
In terms of the distance graph, this implies that $\C_{0}$ has a length which is by at least $\sum_{c\in \mathcal{S}(S')} (o(c)-\gamma\cdot w(c))+\sum_{c\in \mathcal{L}(S')} (o(c)-2\cdot w(c))$ larger than the length of $\C$. The length of $\C$ is $\sum_{c\in \mathcal{S}(S')}w(c) + \sum_{c\in \mathcal{L}(S')}w(c)$.
Therefore, \eqref{eq:secondBoundNew} is then implied by the following sequence of calculations:
\begin{align*}
  |\C_{0}| & \ge \sum_{c\in \mathcal{S}(S')} (o(c)-\gamma\cdot w(c))+\sum_{c\in \mathcal{L}(S')} (o(c)-2\cdot w(c)) +
  \sum_{c\in \mathcal{S}(S')}w(c) + \sum_{c\in \mathcal{L}(S')}w(c)                                                   \\ &=
  \sum_{c\in \mathcal{S}(S')} (o(c)-(\gamma-1)\cdot w(c))+\sum_{c\in \mathcal{L}(S')} (o(c)- w(c))                    \\ & =
  o'- (\gamma-1)\cdot \sum_{c\in \mathcal{S}(S')}  w(c)-\sum_{c\in \mathcal{L}(S')} w(c)\,,
\end{align*}
and this implies~\eqref{eq:secondBound}, as noted above.

To show the desired lower bound on the difference of total overlap between $\C$ and $\C_0$, we slowly ``transform'' $\C_0$ into $\C$ and track how each step of the transformation increases the total overlap. Next, we describe these individual transformation steps in more detail.

\begin{figure}
  \begin{center}
    \begin{tikzpicture}[node distance=2cm and 2cm,-{Latex[length=2mm]}, vertex/.style={draw,circle,fill=black,inner sep=0pt,minimum size=3pt}]
      \node[vertex] (u) [label={[label distance=-4]above left:$u$}] {};
      \node[vertex] (v) [right of=u, label={[label distance=-4]above right:$v$}] {};
      \node[vertex] (u') [below of=u, label={[label distance=-4]below left:$u'$}] {};
      \node[vertex] (v') [below of=v, label={[label distance=-4]below right:$v'$}] {};
      \draw (u) to node[above] {$e$} (v);
      \draw (u) to node[left] {$f'$} (u');
      \draw (v') to node[right] {$f$}  (v);
      \draw (v') to node[below] {$e'$} (u');
    \end{tikzpicture}
    \caption{Illustration of the notation for $\swap(\overline{\C},e)$. Note that we also allow nodes to be equal to one another here, e.g., it could be that $u=v$, in which case $e$ is a self-loop.}
    \label{fig:second_bound_notation}
  \end{center}
\end{figure}
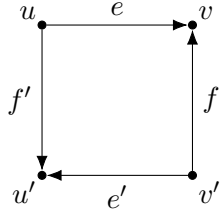

Consider any \CC $\overline{\C}$ and a directed edge $e=(u,v)$ which is not contained in $\overline{\C}$ (note that $u=v$ is possible because the graphs contain self-loops). Then we can modify $\overline{\C}$ slightly such that it does contain $e$. Specifically, let $f=(v',v)$ be the incoming edge of $v$ in $\overline{\C}$ and $f'=(u,u')$ be the outgoing edge of $u$ in $\overline{\C}$. Then, we can add $e$ and $e'=(v',u')$ to $\overline{\C}$ and instead remove $f$ and $f'$ from $\overline{\C}$. The resulting set of edges forms a \CC $\overline{\C}'$ which now includes the edge $e$. We call this operation an \emph{edge swap}. Note that the edge swap is completely determined by the given \CC $\overline{\C}$ and the edge $e$. We refer to this unique swap as $\swap(\overline{\C},e)$ and always refer to the edges that are added to the \CC as $e$ and $e'$ and to the edges which are removed as $f$ and $f'$;
see Figure~\ref{fig:second_bound_notation} for an illustration of the notation.

Given a \CC $\C_0$ (in our case the maximum Hamiltonian cycle) and the \CC $\C$, we can transform $\C_0$ into $\C$ by a sequence of edge swaps. Specifically, if $\C_i$ is a \CC, we can take any edge $e\in \C\setminus \C_i$, i.e., any edge in $\C$ that is not in $\C_i$, and obtain a new \CC $\C_{i+1}$ from $\C_i$ by performing $\swap(\C_i,e)$. Note that because $e \in \C$, the edges $f$ and $f'$ which are swapped out in $\swap(\C_i,e)$ cannot be part of $\C$. If $e'$ belongs to $\C$, the symmetric difference between $\C_{i+1}$ and $\C$ contains four fewer edges than the symmetric difference between $\C_{i}$ and $\C$ (namely all four edges $e$, $e'$, $f$, and $f'$). If $e'$ is not part of $\C$, the symmetric difference between $\C_{i+1}$ and $\C$ contains two fewer edges than the symmetric difference between $\C_{i}$ and $\C$ (it no longer contains $e$, $f$, and $f'$, but it now contains $e'$). In either case, the number of edges in the symmetric difference always decreases and therefore, after a finite number of such edge swap operations, we obtain a \CC $\C_{\ell}$ which is identical to $\C$.

If we obtain $\C_{i+1}$ from $\C_i$ by swapping in the edges $e$ and $e'$ and swapping out the edges $f$ and $f'$, then the total overlap of $\C_{i+1}$ is larger than the total overlap of $\C_i$ by $|\ov(e)|+|\ov(e')|-|\ov(f)|-|\ov(f')|$.

For a \CC $\C_i$, let $\mathcal{M}(\C_i)$ be the set of small cycles of $\C$ which are also part of $\C_i$.
In other words, if $\C_i$ contains a self-loop $(s,s)$ and the string $s$ corresponds to a small cycle $c$ in $\C$, then (and only then) $c \in \mathcal{M}(\C_i)$.
Note that since $\swap(\C_i,e)$ for $e\in \C\setminus \C_i$ only removes edges $f, f'\in \C_{i}\setminus\C$ from $\C_{i}$,
it holds that $\mathcal{M}(\C_{i+1})\supseteq \mathcal{M}(\C_i)$.

Ideally, we would want to show that we can always choose an edge $e \in \C \setminus \C_i$ such that the total overlap increase from $\C_i$ to $\C_{i+1}$ is at least $\sum_{c \in \mathcal{M}(\C_{i+1})\setminus \mathcal{M}(\C_i)} (o(c)- \gamma\cdot w(c))$. It would not be difficult to see that summing over all $i$ would then imply the desired result, i.e., inequality~\eqref{eqn:overlap-difference}, even without the sum over large cycles. Unfortunately, this appears difficult and in some cases we have to allow for slightly smaller increases. To address this, we relate some small cycles and some large cycles to one another.

We define a relation $T$ between small cycles and a large cycle as follows. A small cycle $c$ of $\C$ and a large cycle $c'$ of $\C$ are \emph{related} if $(\gamma-2)\cdot w(c)\le w(c')$ and the large cycle has a string $s'$ such that $|\ov(s,s')|\ge \alpha\cdot w(c')$ or $|\ov(s',s)|\ge \alpha\cdot w(c')$, where $s$ is the only string corresponding to the small cycle, by the input modification in Section~\ref{sec:modifyingInput}. In this case, and only in this case, we have $(c,c')\in T$.

\begin{lemma}\label{lem:related-count}
  For every large cycle $c'$ of $\C$, at most two different small cycles of $\C$ are related to $c'$.
\end{lemma}
\begin{proof}
  Suppose for a contradiction that there are three small cycles $c_1, c_2$, and $c_3$
  related to cycle $c'$.
  For $j\in \{1, 2, 3\}$, let $s_j$ be the only string of cycle $c_j$
  and let $o_j$ be the overlap from the definition of the relation satisfying $|o_j| \ge \alpha\cdot w(c')$, i.e., either $o_j = \ov(s_j, s'_j)$ or $o_j = \ov(s'_j, s_j)$ for some string $s'_j$ from $c'$.
  Note that since $o_j$ is a suffix or prefix of $s_j$ (depending on whether $o_j = \ov(s_j, s'_j)$ or $o_j = \ov(s'_j, s_j)$), Corollary~\ref{cor:lm:maOverlap} implies
  \begin{equation}\label{eqn:related-count overlap bound}
    |\ov(o_1, o_2)| < w(c_1) + w(c_2) \le \frac{2}{\gamma-2}\cdot w(c')\,,
  \end{equation}
  where the second inequality holds as both $c_1$ and $c_2$ are related to $c'$.
  Using the same argument, both $|\ov(o_2, o_3)|$ and $|\ov(o_3, o_1)|$ are also strictly smaller
  than $\frac{2}{\gamma-2}\cdot w(c')$.

  Each overlap string $o_j$ appears as substring in the semi-infinite string $s(c')^\infty$ for the large cycle $c'$, since each $s'_j$ is a substring of $s(c')^\infty$ by Lemma~\ref{lm:Claim2Blum}.
  For $j\in \{1, 2, 3\}$, let $i_j\in [1,w(c')]$ be the smallest index such that $o_j$ is a prefix of $s(c')^\infty[i_j]$.
  W.l.o.g., suppose that $i_1\le i_2\le i_3$ (by reordering indexes of $c_1, c_2$, and $c_3$).
  Observe that
  $$i_2 - i_1 > \left(\alpha - \frac{2}{\gamma-2}\right) w(c')\,,$$
  since otherwise, $o_1$ and $o_2$ would overlap by at least $\frac{2}{\gamma-2}w(c')$ (using that $o_1$ and $o_2$ have length at least $\alpha\cdot w(c')$), contradicting~\eqref{eqn:related-count overlap bound}.
  Similarly, it holds that $i_3 - i_2 >  \left(\alpha - \frac{2}{\gamma-2}\right) w(c')$ and $i_1+w(c') - i_3 >  \left(\alpha - \frac{2}{\gamma-2}\right) w(c')$;
  for the latter, we use that $o_1$ is also a prefix of $s(c')^\infty[i_1 + w(c')]$ as $|s(c')| = w(c')$ is a periodicity of $s(c')^\infty$.
  Finally, we get a contradiction as follows:
  $$w(c') =  (i_2 - i_1) + (i_3 - i_2) + (i_1+w(c') - i_3) > 3\cdot  \left(\alpha - \frac{2}{\gamma-2}\right)\cdot w(c') \ge w(c')\,,$$
  where the last step uses~\eqref{eqn:alpha_gamma_constr_related_cycles}.
\end{proof}

With this we define $$\Delta_i = \sum_{c \in \mathcal{M}(\C_{i+1})\setminus \mathcal{M}(\C_i)} \Bigg(o(c)- \gamma \cdot w(c) - \frac12\cdot\sum_{c' : (c,c')\in T} (2\cdot w(c')-o(c'))\Bigg) \;.$$

We will show that, for every $i$, we can choose $e \in \C\setminus \C_i$ such that the total overlap increase from $\C_i$ to $\C_{i+1}$ is at least $\Delta_i$ when we obtain $\C_{i+1}$ from $\C_i$ by performing $\swap(\C_i,e)$. Note that the value of $\Delta_i$ does depend on $\C_{i+1}$ and therefore on the edge $e$ that we choose.

Summing over all $i$ gives the desired result since then the total overlap increase is at least
\begin{align*}
  \sum_{i=0}^{\ell-1} \Delta_i & = \sum_{c \in \mathcal{M}(\C_{\ell})\setminus \mathcal{M}(\C_0)} \Bigg(o(c)- \gamma \cdot w(c)- \frac12\cdot\sum_{c' : (c,c')\in T} (2\cdot w(c')-o(c'))\Bigg) \\ &= \sum_{c \in \mathcal{S}(S')} \Bigg(o(c)- \gamma \cdot w(c)- \frac12\cdot\sum_{c' : (c,c')\in T} (2\cdot w(c')-o(c'))\Bigg)
  \\ &\ge \sum_{c \in \mathcal{S}(S')} \Big(o(c)- \gamma \cdot w(c)\Big)- 2\cdot \frac12\cdot \sum_{c' \in \mathcal{L}(S')} \Big(2\cdot w(c')-o(c')\Big)\\ &=
  \sum_{c \in \mathcal{S}(S')} \Big(o(c)- \gamma \cdot w(c)\Big) + \sum_{c' \in \mathcal{L}(S')} \Big(o(c')-2\cdot w(c')\Big)
\end{align*}
and this is what we wanted in (\ref{eqn:overlap-difference}). Here, the first line follows because $\mathcal{M}(\C_{i+1}) \supseteq \mathcal{M}(\C_{i})$ for all $i$ as noted above, the second line follows because $\C_{\ell}=\C$ and $\mathcal{M}(\C_0)=\emptyset$, and the third line follows from Lemma~\ref{lem:related-count}.
Strictly speaking, it is possible that $\mathcal{M}(\C_0)\neq \emptyset$. However, $\C_0$ is a Hamiltonian cycle, and therefore, the only case in which this happens is if this Hamiltonian cycle is in fact a single small cycle $c$, in which case, by Observation~\ref{obs:smallCyclesRotations}, \GREEDY computes an optimal solution.

We will sometimes use the fact that the term $2 w(c')-o(c')$ is non-negative for every large cycle $c'$. Therefore, the part of the definition of $\Delta_i$ that sums over large cycles $c'$ such that $c$ is related to $c'$ can only decrease the value of $\Delta_i$ (and makes it easier to find a suitable edge $e$ in some cases), i.e.,
\begin{equation}\label{eqn:Delta_i_UB}
  \Delta_i \le \sum_{c \in \mathcal{M}(\C_{i+1})\setminus \mathcal{M}(\C_i)} \Bigg(o(c)- \gamma \cdot w(c)\Bigg) \;.
\end{equation}

In Section~\ref{sec:2ndBound-analysis}, we will show that for any \CC $\C_i \neq \C$, it is always possible to find an edge $e\in\C\setminus\C_i$ such that if we obtain $\C_{i+1}$ by performing the $\swap(\C_i,e)$, the total overlap increase is at least $\Delta_i$.
Before that, we present three useful lemmas.

\subsection{Useful Lemmas}

Tarhio and Ukkonen~\cite{tarhio} and Turner~\cite{Turner89} show the following lemma.

\begin{lemma}\label{lem:monge}
  Let $e=(u,v)$, $f=(v',v)$, $f'=(u,u')$, and $e'=(v',u')$ be edges in $G_{\ov}(S')$ such that $\max\{|\ov(e)|,\ov(e')|\}\ge \max\{|\ov(f)|, |\ov(f')|\}$. Then $|\ov(e)|+|\ov(e')|-|\ov(f)|-|\ov(f')|\geq 0$.
\end{lemma}

The following is a slightly different, but somewhat related inequality which gives us better bounds when $e$ is the (only) edge of a small cycle in $\C$. Another difference to Lemma~\ref{lem:monge} is that the following lemma can also be applied if $\max\{|\ov(e)|,|\ov(e')|\} < \max\{|\ov(f)|, |\ov(f')|\}$.

\begin{lemma}\label{lem:small-monge}
  Let $e=(u,v)$, $f=(v',v)$, $f'=(u,u')$, and $e'=(v',u')$ be edges in $G_{\ov}(S')$ such that $e$ is an edge in a small cycle $c$ in $\C$. Then
  \[
    |\ov(e)|+|\ov(e')|- |\ov(f)|-|\ov(f')|> |\ov(e)|-\max\{|\ov(f)|,|\ov(f')|\}-w(c) \;.
  \]
\end{lemma}
\begin{proof}
  If $\min\{|\ov(f)|,|\ov(f')|\} < w(c)$, then trivially $|\ov(e)|+|\ov(e')|- |\ov(f)|-|\ov(f')|\ge |\ov(e)|- |\ov(f)|-|\ov(f')|>|\ov(e)|-\max\{|\ov(f)|,|\ov(f')|\}-w(c)$ and we are done.
  So now assume $\min\{|\ov(f)|,|\ov(f')|\} \ge w(c)$.

  First note that since $e$ is an edge of a small cycle in $\C$, $e$ is a self-loop in $G_{\ov}(S')$ and $u=v$.
  Since $\ov(f)$ is a prefix of $u=v$, we observe that $\ov(f) = u[1,|\ov(f)|]$. Because $u$ has period $w(c)$ by Lemma~\ref{lm:equivSmall}, this also implies
  $\ov(f) = u[1+k\cdot w(c),|\ov(f)|+k\cdot w(c)]$, where we choose $k\ge0$ as the largest integer for which $k\cdot w(c) \le |u|-\max\{|\ov(f)|,|\ov(f')|\}$.  For this choice of $k$, we have $k\cdot w(c) > |u|-\max\{|\ov(f)|,|\ov(f')|\} - w(c)$.

  Furthermore, $\ov(f') = u[|u|-|\ov(f')|+1,|u|]$ because $\ov(f')$ is a suffix of $u$. Hence, the string $u[|u|-|\ov(f')| +1, |\ov(f)|+k\cdot w(c)]$ is a suffix of $\ov(f)$ as well as a prefix of $\ov(f')$. This string has length $|\ov(f)|+k\cdot w(c) - (|u|-|\ov(f')|)
    > |\ov(f)| + |u|-\max\{|\ov(f)|,|\ov(f')|\} - w(c) - (|u|-|\ov(f')|) = \min\{|\ov(f)|,|\ov(f')|\} - w(c)$, which is non-negative by the assumption above.

  Every suffix of $\ov(f)$ is also a suffix of $v'$ and every prefix of $\ov(f')$ is also a prefix of $u'$. Hence, $v'$ has a suffix of length larger than $\min\{|\ov(f)|,|\ov(f')|\} - w(c)$ which is identical to a prefix of $u'$. Therefore, $|\ov(e')| >  \min\{|\ov(f)|,|\ov(f')|\} - w(c)$, which implies the lemma.
\end{proof}

Under a certain condition, we can further strengthen the inequality of the previous lemma.

\begin{lemma}\label{lm:lastR}
  Consider the edges $e=(u,v)$, $f=(v',v)$, $f'=(u,u')$, and $e'=(v',u')$. Suppose $e$ is an edge in a (large or small) cycle $c$ of $\C$, $e'$ is an edge in a (large or small) cycle $c'$ of $\C$, and $|\ov(e')| \ge w(c)+w(c')$. Then
  \[
    |\ov(e)|+|\ov(e')|- |\ov(f)|-|\ov(f')|> |\ov(e)|-w(c) \;.
  \]
\end{lemma}

\begin{proof}
  We show that $|\ov(e')|>|\ov(f)|+|\ov(f')|-w(c)$, which trivially implies the lemma.

  If $\min\{|\ov(f)|,|\ov(f')|\} \le w(c)$, this inequality holds because by using Lemma~\ref{lm:maOverlap}, we get $|\ov(e')|\ge w(c)+w(c') > \max\{|\ov(f)|,|\ov(f')|\} \ge \max\{|\ov(f)|,|\ov(f')|\} + \min\{|\ov(f)|,|\ov(f')|\}-w(c)=|\ov(f)|+|\ov(f')|-w(c)$. Hence, for the remainder of the proof, we assume that we have $\min\{|\ov(f)|,|\ov(f')|\} > w(c)$.

  Now, assume for a contradiction that $|\ov(e')|\le|\ov(f)|+|\ov(f')|-w(c)$. We claim that in this case $\ov(e')$ has a periodicity of $w(c)$, i.e., $\ov(e')$ is a prefix of $x^\infty$ for some string $x$ with $|x| = w(c)$. To show this, recall that $|\ov(e')| \ge w(c)+w(c') > \max\{|\ov(f')|, |\ov(f)|\}$ by Lemma~\ref{lm:maOverlap}.
  Since $\ov(f')$ is a prefix of $u'$ and a suffix of $u$ and since $\ov(e')$ is a prefix of $u'$,
  the first $|\ov(f')|$ characters of $\ov(e')$ are also a suffix of $u$, i.e., $$\ov(e')[1,|\ov(f')|] = \ov(f') = u[|u|-|\ov(f')|+1,|u|]\,.$$
  Similarly, since $\ov(f)$ is a prefix of $v$ and a suffix of $v'$ and since $\ov(e')$ is a suffix of $v'$, we get that $$\ov(e')[|\ov(e')|-|\ov(f)|+1,|\ov(e')|] = \ov(f) = v[1, |\ov(f)|]\,.$$
  Observe that for all $1\le i \le |\ov(e')|-w(c)$, a character at position $i$ of $\ov(e')$ must be the same as the character at position $i+w(c)$ of $\ov(e')$. Indeed, if $i+w(c)\le |\ov(f')|$, this is true as $u$ has a periodicity of $w(c)$. If $i>|\ov(e')|-|\ov(f)|$, it is true because $v$ has periodicity $w(c)$. One of these two cases must apply because otherwise, $i+w(c)> |\ov(f')|$ and $i\le|\ov(e')|-|\ov(f)|$, which implies $|\ov(f')|-w(c)<i\le |\ov(e')|-|\ov(f)|$, contradicting our assumption that $|\ov(f')|+|\ov(f)| \ge |\ov(e')| + w(c)$. Hence, $\ov(e')$ has a periodicity of $w(c)$ (in particular, $\period(\ov(e'))\le w(c)$).

  Next, we show that $\ov(e')$ is a substring of the semi-infinite string $s(c)^{\infty}$. Because $\ov(e')$ has a periodicity of $w(c)$ and $s(c)^{\infty}$ has period $w(c)$, it is sufficient to argue that the first $w(c)$ characters of $\ov(e')$ are a substring of $s(c)^{\infty}$. This is indeed the case since $\ov(e')[1,|\ov(f')|]$ is a substring of $u$ which is a substring of $s(c)^{\infty}$ and we assumed that $|\ov(f')|> w(c)$.

  Since $\ov(e')$ is a substring of $s(c)^{\infty}$ as well as of $s(c')^{\infty}$ (because $e'$ lies on cycle $c'$), Corollary~\ref{cor:lm:maOverlap} implies $|\ov(e')|<w(c)+w(c')$ which contradicts the assumption of the lemma.
\end{proof}

\subsection{Analysis}\label{sec:2ndBound-analysis}

In this section, we will show that for any \CC $\C_i \neq \C$, it is always possible to find an edge $e\in\C\setminus\C_i$ such that if we obtain $\C_{i+1}$ by performing the $\swap(\C_i,e)$, the total overlap increase is at least $$\Delta_i = \sum_{c \in \mathcal{M}(\C_{i+1})\setminus \mathcal{M}(\C_i)} \Bigg(o(c)- \gamma \cdot w(c) - \frac12\cdot\sum_{c' : (c,c')\in T} (2\cdot w(c')-o(c'))\Bigg)\,.$$

The following defines the concept of a \emph{good edge}. It is a slightly technical definition, but it is useful in the sense that (a) we will be able to show that a good edge $e$ is always a suitable choice for $\swap(\C_i,e)$ and (b) in many cases we can find a good edge. For the remaining cases (i.e., when it is not obvious whether a good edge exists), we will have separate arguments that show that an appropriate swap is possible.

\begin{definition}\label{def:goodedge}
  We call an edge $e=(u,v) \in \C \setminus \C_i$ a \emph{good} edge if the following statements hold for the $\swap(\C_i,e)$ which swaps out edges $f=(v',v) \in \C_i \setminus \C$ and $f'=(u,u') \in \C_i \setminus \C$ and swaps in edges $e=(u,v)$ and $e'=(v',u')$:
  \begin{itemize}
    \item $e$ belongs to a small cycle $c$ of $\C$ and $e'$ does not belong to a small cycle of $\C$.
    \item If $|\ov(f)| \ge |\ov(f')|$, then for the cycle $c'$ in $\C$ that contains the string $v'$, it holds that either $|\ov(f)| \ge o(c')$ or $c'$ is a small cycle with $w(c') \le w(c)$.
    \item If $|\ov(f')| > |\ov(f)|$,   then for the cycle $c'$ in $\C$ that contains the string $u'$, it holds that either $|\ov(f')| \ge o(c')$ or $c'$ is a small cycle with $w(c') \le w(c)$.
  \end{itemize}
\end{definition}

\begin{figure}
  \centering
  \begin{subfigure}[b]{0.475\textwidth}
    \centering
    \begin{tikzpicture}[-{Latex[length=2mm]}, vertex/.style={draw,circle,fill=black,inner sep=0pt,minimum size=3pt}, decoration={
          markings,
          mark=between positions 0.125 and 0.875 step 0.125 with {\node[inner sep=0pt, minimum size=2pt](y\pgfkeysvalueof{/pgf/decoration/mark info/sequence number}) {};},
          mark=between positions 0.125 and 0.875 step 0.125 with {\fill (0pt,0pt) circle (2pt);},
        }]
      % \draw[help lines,step=8pt] (-3.3, 1.4) grid (4.3, -4);
      \path[use as bounding box] (-3.3, 1.4) rectangle (4.3, -4);
      \node[vertex] (u) at (0,0) [label={[label distance=5]right:$u=v$}] {};
      \node[vertex] (u') at (-2,-2) [label={[label distance=-4]below left:$u'$}] {};
      \node[vertex] (v') at (2,-2) [label={[label distance=-4]below right:$v'$}] {};
      %        \node[vertex] (x) at (2,-4) {};
      \draw (u) to node[left, outer sep=5, pos=0.3] {$f'\in\C_i\setminus \C$} (u');
      \draw (v') to node[right, outer sep=5, pos=0.7] {$f\in\C_i\setminus \C$}  (u);
      \draw (v') [dotted] to node[above] {$e'$} (u');
      \path (u) edge [in=20,out=160,loop,looseness=70] node[above] {$e\in\C\setminus \C_i$} node[below, outer sep=4pt] {$c$} (u);
      %        \draw (v') to node[left] {} (x);
      \path (v') edge [-,decorate, out=270, in=10,loop,looseness=130] node[left, outer sep=20pt, pos=0.6] {$c'$} (v');
      \draw (v') edge (y1) ++(y1) edge (y2) ++(y2)edge (y3) ++(y3)edge (y4) ++(y4)edge (y5) ++(y5)edge (y6) ++(y6) edge (y7) ++(y7)  edge(v');
    \end{tikzpicture}
  \end{subfigure}
  \hfill
  \begin{subfigure}[b]{0.475\textwidth}
    \centering
    \begin{tikzpicture}[-{Latex[length=2mm]}, vertex/.style={draw,circle,fill=black,inner sep=0pt,minimum size=3pt}, decoration={
          markings,
          mark=between positions 0.125 and 0.875 step 0.125 with {\node[inner sep=0pt, minimum size=2pt](y\pgfkeysvalueof{/pgf/decoration/mark info/sequence number}) {};},
          mark=between positions 0.125 and 0.875 step 0.125 with {\fill (0pt,0pt) circle (2pt);},
        }]
      % \draw[help lines,step=8pt] (-3.3, 1.4) grid (4.3, -4);
      \path[use as bounding box] (-3.3, 1.4) rectangle (4.3, -4);
      \node[vertex] (u) at (0,0) [label={[label distance=5]right:$u=v$}] {};
      \node[vertex] (u') at (-2,-2) [label={[label distance=-4]below left:$u'$}] {};
      \node[vertex] (v') at (2,-2) [label={[label distance=-4]below right:$v'$}] {};
      %        \node[vertex] (x) at (2,-4) {};
      \draw (u) to node[left, outer sep=5, pos=0.3] {$f'\in\C_i\setminus \C$} (u');
      \draw (v') to node[right, outer sep=5, pos=0.7] {$f\in\C_i\setminus \C$}  (u);
      \draw (v') [dotted] to node[above] {$e'$} (u');
      \path (u) edge [in=20,out=160,loop,looseness=70] node[above] {$e\in\C\setminus \C_i$} node[below, outer sep=4pt] {$c$} (u);
      %        \draw (v') to node[left] {} (x);
      \path (v') edge [out=270, in=10,loop,looseness=130] node[left, outer sep=20pt, pos=0.6] {$c'$} (v');
      %\draw (v') edge (y1) ++(y1) edge (y2) ++(y2)edge (y3) ++(y3)edge (y4) ++(y4)edge (y5) ++(y5)edge (y6) ++(y6) edge (y7) ++(y7)  edge(v');
    \end{tikzpicture}
  \end{subfigure}
  \caption{Illustration of a good edge $e$ for different cases. The edge $e'$ is not allowed to be in a small cycle of $\C$. It can either not be contained in $\C$ at all or it can be part of a large cycle of $\C$. For these illustrations, we also assume that $|\ov(f)|\ge |\ov(f')|$. On the left, $c'$ is a large cycle and $|\ov(f)|$ is at least $o(c')$ (it is possible that $e'$ is part of the large cycle $c'$). On the right, $c'$ is a small cycle and $w(c)\ge w(c')$.}
  \label{fig:good-edge}
\end{figure}
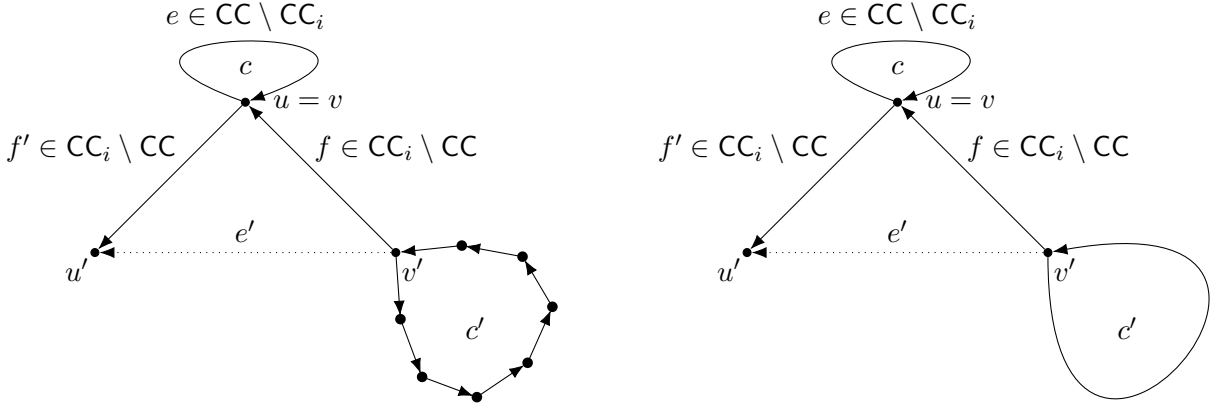

The following lemma shows that if there is a good edge $e$, performing $\swap(\C_i,e)$ results in a sufficient increase of the total overlap.
\begin{lemma}\label{lem:goodedge}
  If $e$ is a good edge, then after performing $\swap(\C_i,e)$, the resulting \CC $\C_{i+1}$ has by at least $\Delta_i$ larger total overlap than $\C_i$.
\end{lemma}
\begin{proof}
  By definition of a good edge, $e$ is the edge of a small cycle $c$.
  Due to Lemma~\ref{lem:small-monge}, $|\ov(e)|+|\ov(e')|-|\ov(f)|-|\ov(f')|> |\ov(e)|-\max\{|\ov(f)|,|\ov(f')|\} - w(c)$.

  Suppose $|\ov(f)| \ge |\ov(f')|$ (the other case is analogous) and let $c'$ be the cycle containing the string $v'$.
  Then $|\ov(e)|-\max\{|\ov(f)|,|\ov(f')|\} - w(c) = |\ov(e)|-|\ov(f)| - w(c) > |\ov(e)|-w(c)-w(c') - w(c) = |\ov(e)|-2w(c)-w(c')$, where the inequality follows from Lemma~\ref{lm:maOverlap}.
  Hence, it is sufficient to show that $|\ov(e)|-2w(c)-w(c') \ge \Delta_i$.

  Since $e$ is the edge of a small cycle in $\C$ and $e'$ is not an edge of a small cycle in $\C$ (by the definition of a good edge), if we obtain $\C_{i+1}$ from $\C_i$ by performing $\swap(\C_i,e)$, then $\mathcal{M}(\C_{i+1})\setminus \mathcal{M}(\C_i) = \{c\}$. In this case,
  \begin{align}
    \Delta_i & = o(c)- \gamma \cdot w(c) - \frac12\cdot\sum_{c'' : (c,c'')\in T} (2\cdot w(c'')-o(c'')) \nonumber \\ &= |\ov(e)|- \gamma \cdot w(c) -
    \frac12\cdot\sum_{c'' : (c,c'')\in T} (2\cdot w(c'')-o(c''))                                     \nonumber    \\ &\le
    \begin{cases}
      |\ov(e)|- \gamma \cdot w(c) - \frac12\cdot (2\cdot w(c')-o(c')) & \text{if } (c,c')\in T \\
      |\ov(e)|- \gamma \cdot w(c)                                     & \text{otherwise}
    \end{cases}                                                                         \nonumber  \\
             & \le
    \begin{cases}
      |\ov(e)|- \gamma \cdot w(c) - \frac12\cdot (w(c')-w(c)) & \text{if } (c,c')\in T \\
      |\ov(e)|- \gamma \cdot w(c)                             & \text{otherwise}
    \end{cases}\;. \label{eqn:goodedge_Delta_i_UB}
  \end{align}
  The last step follows since if $(c,c')\in T$, then $c'$ is a large cycle and therefore, $o(c')\le |\ov(f)| < w(c)+w(c')$, where the first inequality follows from the definition of a good edge and the last inequality follows from Lemma~\ref{lm:maOverlap}.

  The following fact establishes an upper bound on $w(c')$ by a function of $w(c)$.
  \begin{fact}\label{fact:largecyclebound}
    \hfill
    \begin{itemize}
      \item If $c'$ is a large cycle and $(c,c')\in T$, then $w(c') < \frac{1}{\alpha-1} w(c)$.
      \item Otherwise, $w(c') < (\gamma-2)\cdot w(c)$ holds.
    \end{itemize}
  \end{fact}
  \begin{proof}
    If $c'$ is a large cycle, then $w(c)+w(c') > |\ov(f)| \ge o(c') > \alpha \cdot w(c')$, where the first step follows from Lemma~\ref{lm:maOverlap}, the second step follows from the definition of a good edge, and the last step follows from the definition of a large cycle. Rearranging this inequality gives $w(c') < \frac{1}{\alpha-1} w(c)$.

    Now, to show the second claim, there are two cases. If $c'$ is a large cycle, but $(c,c')\notin T$, then we again recall that $|\ov(f)| \ge \alpha\cdot w(c')$. Since $(c,c')\notin T$, this implies that $w(c') < (\gamma-2)\cdot w(c)$ as claimed.
    On the other hand, if $c'$ is a small cycle,
    then, due to the definition of a good edge, either $w(c')\le w(c)$ or $|\ov(f)| \ge o(c')$. In the former case, we are already done as $\gamma > 3$. In the latter case, $|\ov(f)| \ge o(c')> 2 w(c')$ and hence $w(c) > |\ov(f)|-w(c') > w(c')$, where the first inequality follows from Lemma~\ref{lm:maOverlap}. Again, this implies the second claim as $\gamma > 3$.
  \end{proof}

  Finally, to show that $|\ov(e)| - 2w(c) - w(c')\ge \Delta_i$, we distinguish two cases and utilize the upper bound on $\Delta_i$ derived in~\eqref{eqn:goodedge_Delta_i_UB}.
  \begin{itemize}
    \item If $c'$ is large cycle and $(c,c')\in T$, then using the first claim in Fact~\ref{fact:largecyclebound},
          \begin{align*}
            |\ov(e)| - 2w(c) - w(c')
             & =   |\ov(e)| - 2w(c) - \frac12 w(c') - \frac12 w(c')                                               \\
             & > |\ov(e)| - 2w(c) - \frac12 w(c') - \frac{1}{2(\alpha-1)}  w(c)                                   \\
             & =   |\ov(e)| - \left(2 +\frac{1}{2(\alpha-1)}\right)\cdot  w(c) - \frac12 w(c')                    \\
             & =   |\ov(e)| - \left(\frac52 +\frac{1}{2(\alpha-1)}\right)\cdot w(c) - \frac12 w(c') +\frac12 w(c) \\
             & \ge |\ov(e)| - \gamma\cdot w(c) - \frac12 w(c') +\frac12 w(c)  \ge \Delta_i\,,
          \end{align*}
          where the last line uses~\eqref{eqn:alpha_gamma_constr_good_edge}.
    \item Otherwise, $|\ov(e)| - 2w(c) - w(c') \ge |\ov(e)| - \gamma\cdot w(c) \ge \Delta_i$.\qedhere
  \end{itemize}
\end{proof}

There may be cases where $\C\setminus \C_i$ does not necessarily have a good edge. In such cases, we can use other arguments. The following lemma is an example of this.

\begin{lemma}\label{lem:longcycles}
  If there exists an edge $e\in \C\setminus \C_i$ such that (i) $\swap(\C_i,e)$ swaps in edges $e$ and $e'$, (ii) neither $e$ nor $e'$ are edges of a small cycle in $\C$, and (iii) $\max\{|\ov(e)|,|\ov(e')|\}\ge \max\{|\ov(f)|,|\ov(f')|\}$, then after performing $\swap(\C_i,e)$, the resulting \CC $\C_{i+1}$ has by at least $\Delta_i$ larger total overlap than $\C_i$.
\end{lemma}
\begin{proof}
  If neither $e$ nor $e'$ are edges of a small cycle in $\C$, then performing $\swap(\C_i,e)$ results in a \CC $\C_{i+1}$ for which $\mathcal{M}(\C_{i+1})\setminus \mathcal{M}(\C_i) = \emptyset$. Therefore, $\Delta_i = 0$.
  Since $\max\{|\ov(e)|,|\ov(e')|\}\ge \max\{|\ov(f)|,|\ov(f')|\}$, Lemma~\ref{lem:monge} implies that $|\ov(e)|+|\ov(e')|\ge |\ov(f)| + |\ov(f')|$. Hence, $|\ov(e)|+|\ov(e')|- |\ov(f)| - |\ov(f')| \ge 0 = \Delta_i$.
\end{proof}

If there is an edge $e\in \C\setminus \C_i$ such that performing $\swap(\C_i,e)$ reduces the symmetric difference between $\C$ and $\C_i$ by four, then we show that $\swap(\C_i,e)$ increases the total overlap by at least $\Delta_i$.

\begin{lemma}\label{lem:2ring}
  If there exists an edge $e\in \C\setminus \C_i$ such that performing $\swap(\C_i,e)$ reduces the symmetric difference between the \CC $\C_i$ and $\C$ by four edges, then after performing $\swap(\C_i,e)$, the resulting \CC $\C_{i+1}$ has by at least $\Delta_i$ larger total overlap than $\C_i$.
\end{lemma}
\begin{proof}
  Recall that $\swap(\C_i,e)$ adds the edges $e$ and $e'$ to the \CC $\C_i$ and removes the edges $f$ and $f'$. Thus, if the symmetric difference to $\C$ decreases by four edges, then it must be the case that $e, e' \in \C\setminus \C_i$ and $f, f' \in \C_i \setminus \C$.

  We have $\max\{|\ov(e)|,|\ov(e')|\}\ge\max\{|\ov(f)|,|\ov(f')|\}$, since otherwise, \MGREEDY would have picked the edge of greater overlap between $f$ and $f'$ for inclusion in $\C$, before picking either one of $e$ or $e'$.
  We now consider four cases:

  \begin{itemize}
    \item Suppose $e$ and $e'$ both belong to large cycles in $\C$. Then Lemma~\ref{lem:longcycles} applies and we are done.

    \item Suppose $e$ and $e'$ both belong to small cycles in $\C$. Let these two small cycles be $c$ and $c'$ respectively. If we obtain $\C_{i+1}$ from $\C_i$ by performing $\swap(\C_i,e)$, then $\mathcal{M}(\C_{i+1})\setminus \mathcal{M}(\C_i) = \{c,c'\}$.
          Thus, using~\eqref{eqn:Delta_i_UB} together with $|\ov(e)| = o(c)$, $|\ov(e')| = o(c')$, and $\gamma > 2$, we obtain
          $$\Delta_i < |\ov(e)| - 2w(c) + |\ov(e')| - 2w(c')\,.$$
          Due to Lemma~\ref{lm:maOverlap}, $\max\{|\ov(f)|,|\ov(f')|\} < w(c)+w(c')$. Therefore, $|\ov(e)|+|\ov(e')| - |\ov(f)|-|\ov(f')|\ge |\ov(e)|-2w(c)+|\ov(e')|-2w(c') > \Delta_i$ as claimed.

    \item Suppose $e$ belongs to a small cycle $c$ and $e'$ belongs to a large cycle $c'$ in $\C$.

          We distinguish between three cases:
          \begin{itemize}
            \item If $|\ov(e')| \le \max\{|\ov(f)|,|\ov(f')|\}$, then $e$ is a good edge (note that $o(c')\le |\ov(e')|$ because $e'$ belongs to the cycle $c'$) and we apply Lemma~\ref{lem:goodedge}.

            \item If $w(c)+w(c') \ge |\ov(e')| > \max\{|\ov(f)|,|\ov(f')|\}$, then using Lemma~\ref{lm:maOverlap},
                  \begin{align*}
                    \max\{|\ov(f)|,|\ov(f')|\} < w(c)+w(c')
                     & = w(c) + \gamma\cdot w(c') - (\gamma-1)\cdot w(c')                   \\
                     & \le w(c) + (\gamma-1)\cdot \alpha\cdot w(c') - (\gamma-1)\cdot w(c') \\
                     & \le w(c) + (\gamma-1)\cdot o(c') - (\gamma-1)\cdot w(c')             \\
                     & \le w(c) + (\gamma-1)\cdot |\ov(e')| - (\gamma-1)\cdot w(c')         \\
                     & \le w(c) + (\gamma-1)\cdot w(c) = \gamma\cdot w(c)\,,
                  \end{align*}
                  where the second line uses~\eqref{eqn:alpha_gamma_constr_sym_diff_4}, the third line follows from $c'$ being large,
                  the fourth one from that $o(c')$ is the smallest overlap on cycle $c'$, and the fifth line uses the case condition.
                  Now, the increase in the total overlap when performing $\swap(\C_i,e)$ is at least $|\ov(e)|+|\ov(e')|-|\ov(f)|-|\ov(f')| \ge |\ov(e)|-|\ov(f)| > o(c)-\gamma\cdot w(c) \ge \Delta_i$, where we use~\eqref{eqn:Delta_i_UB} together with $|\ov(e)| = o(c)$ and $\gamma > 1$.

            \item Otherwise, we have $|\ov(e')| > w(c)+w(c')$. From Lemma~\ref{lm:lastR}, it follows that $|\ov(e)|+|\ov(e')|-|\ov(f)|-|\ov(f')|\ge  |\ov(e)|-w(c) \ge \Delta_i$.

          \end{itemize}

    \item Suppose $e$ belongs to a large cycle and $e'$ belongs to a small cycle in $\C$. Observe that $\swap(\C_i,e')$ results in exactly the same \CC $\C_{i+1}$ as $\swap(\C_i,e)$. Therefore, we just apply the previous argument to $\swap(\C_i,e')$, and we are done.\qedhere
  \end{itemize}
\end{proof}

Lastly, if neither of the previous two lemmas applies, we can find a good edge for sure:
\begin{lemma}\label{lem:centraledge}
  Suppose Lemmas \ref{lem:longcycles} and \ref{lem:2ring} do not apply, i.e., no edge with the corresponding properties exists.
  Then there exists a good edge in $\C \setminus \C_i$.
\end{lemma}
\begin{proof}
  Let $f_{\max}$ be an edge of $\C_i\setminus \C$ that has the maximum overlap among the edges of $\C_i\setminus \C$. We will show that $f_{\max}$ is a candidate for either $f$ or $f'$.

  Let $e_h$ be the edge of $\C$ that has the same head node as $f_{\max}$ and let $e_t$ be the edge of $\C$ that has the same tail node as $f_{\max}$. We will later pick one of these as our edge $e$.
  We have $|\ov(f_{\max})| \le \max\{|\ov(e_h)|,|\ov(e_t)|\}$ as otherwise, \MGREEDY would have picked edge $f_{\max}$ for inclusion in $\C$ before picking either one of $e_h$ or $e_t$.

  We first show that $e_h$ or $e_t$ satisfies the first condition of a good edge in Definition~\ref{def:goodedge}.

  \begin{fact}\label{fact:goodedge-first-condition}
    \mbox{}
    \begin{itemize}
      \item If $|\ov(e_h)|\ge|\ov(f_{\max})|$, then $e=e_h$ satisfies the first condition of a good edge.
      \item Similarly, if $|\ov(e_t)|\ge|\ov(f_{\max})|$, then $e=e_t$ satisfies the first condition of a good edge.
    \end{itemize}
  \end{fact}
  \begin{proof}
    \mbox{}
    \begin{itemize}
      \item To see that $e=e_h$ satisfies the first condition of a good edge if $|\ov(e_h)|\ge|\ov(f_{\max})|$, consider $\swap(\C_i,e_h)$ and use the same notation as in Figure~\ref{fig:second_bound_notation}.

            First of all, in this case, $f=f_{\max}=(v',v)$ and because $f_{\max}$ was chosen to have the maximum overlap in $\C_i \setminus \C$, $|\ov(f)| \ge |\ov(f')|$. We conclude that $|\ov(e)| \ge \max\{|\ov(f)|, |\ov(f')|\}$.
            If $e$ and $e'$ both belong to $\C$, Lemma~\ref{lem:2ring} applies. Since we assume that the lemma does not apply and since we know that $e\in \C$, it follows that $e'\notin \C$.
            If $e$ belongs to a large cycle in $\C$, Lemma~\ref{lem:longcycles} applies because $e'\notin \C$ and $|\ov(e)| \ge \max\{|\ov(f)|, |\ov(f')|\}$. Because we assume that the lemma does not apply, we conclude that $e$ must belong to a small cycle. Together with $e'\notin\C$, this satisfies the first condition of a good edge.
      \item By symmetric arguments to the above, it also follows that if $|\ov(e_t)|\ge|\ov(f_{\max})|$, then $e=e_t$ satisfies the first condition of a good edge.\qedhere
    \end{itemize}
  \end{proof}

  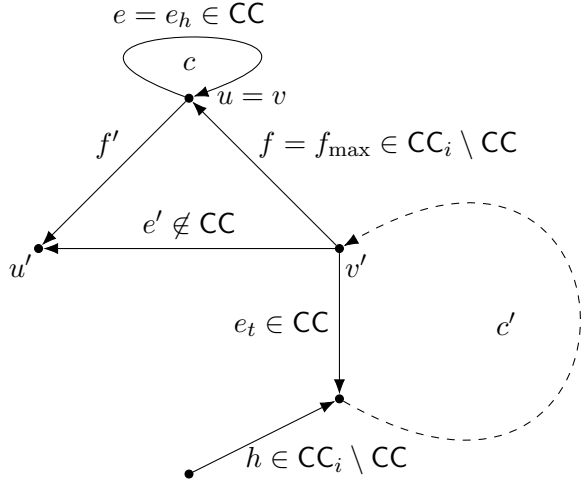
\begin{figure}
    \begin{center}
      \begin{tikzpicture}[-{Latex[length=2mm]}, vertex/.style={draw,circle,fill=black,inner sep=0pt,minimum size=3pt}]
        %\draw[help lines,step=8pt] (-2.5, 1.4) grid (6.1, -5.2);
        \path[use as bounding box] (-6.1, 1.4) rectangle (6.1, -5.2);
        \node[vertex] (u) at (0,0) [label={[label distance=5]right:$u=v$}] {};
        \node[vertex] (u') at (-2,-2) [label={[label distance=-4]below left:$u'$}] {};
        \node[vertex] (v') at (2,-2) [label={[label distance=-4]below right:$v'$}] {};
        \node[vertex] (x) at (2,-4) {};
        \node[vertex] (y) at (0,-5) {};
        \draw (u) to node[left, outer sep=5, pos=0.3] {$f'$} (u');
        \draw (v') to node[right, outer sep=5, pos=0.7] {$f=f_{\max}\in\C_i\setminus \C$}  (u);
        \draw (v') to node[above] {$e'\not\in\C$} (u');
        \path (u) edge [in=20,out=160,loop,looseness=70] node[above] {$e=e_h\in\C$} node[below, outer sep=4pt] {$c$} (u);
        \draw (v') to node[left] {$e_t\in\C$} (x);
        \path (x) edge [bend right, out=240, in=300,looseness=6,dashed] node[left, outer sep=20pt] {$c'$} (v');
        \draw (y) to node[right, outer sep=7, pos=0.17] {$h\in \C_i\setminus \C$} (x);

      \end{tikzpicture}
      \caption{Illustration of Case A in the proof of Lemma~\ref{lem:centraledge}.}
      \label{fig:central_edge-case1}
    \end{center}
  \end{figure}

  To show that we can also satisfy the second or the third condition (for an edge that satisfies the first), we distinguish three cases:

  \smallskip\noindent
  \textbf{Case A:} Suppose $|\ov(e_h)| \ge |\ov(f_{\max})|$ and $|\ov(e_t)| \ge |\ov(f_{\max})|$.

  Let $c$ be the cycle of $\C$ to which $e_h$ belongs and let $c'$ be the cycle of $\C$ to which $e_t$ belongs; see Figure~\ref{fig:central_edge-case1} for an illustration.
  We assume that $w(c)\ge w(c')$ as the arguments for the other case are completely symmetric with the roles of $e_h$ and $e_t$ reversed.

  We claim that $e=e_h$ is a good edge. It follows from Fact~\ref{fact:goodedge-first-condition} that $e$ satisfies the first condition of a good edge.
  Since $|\ov(f)| \ge |\ov(f')|$ as $f = f_{\max}$, it only remains to show the second condition.
  Since $e_t$ is an edge in the cycle $c'$ in $\C$, we have $|\ov(e_t)|\ge o(c')$.
  If $o(c')\le|\ov(f)|$, the second condition of a good edge is already satisfied. So suppose $o(c') > |\ov(f)|$.

  Assume for a contradiction that $c'$ is a large cycle in $\C$. Then consider the edge $h$ in $\C_i \setminus \C$ that has the same head node as $e_t$. We know that $|\ov(f)| \ge |\ov(h)|$ because $f_{\max}=f$ was chosen to have the maximum overlap among all edges in $\C_i \setminus \C$. Hence, $|\ov(e_t)|\ge o(c') > |\ov(f)|\ge |\ov(h)|$ and thus $|\ov(e_t)| > \max\{|\ov(f)|, |\ov(h)|\}$.
  Consider $\swap(\C_i, e_t)$, i.e., with edge $e_t$ acting as edge $e$ in the operation.
  If $\swap(\C_i, e_t)$ reduces the symmetric difference between $C_i$ and $\C$ by four edges, then Lemma~\ref{lem:2ring} applies. Otherwise, $e'\not\in \C$, so Lemma~\ref{lem:longcycles} applies as the cycle $c'$ containing $e = e_t$ is large.
  This is a contradiction to our assumption that neither Lemma~\ref{lem:2ring} nor Lemma~\ref{lem:longcycles} can be applied.

  Thus, $c'$ must be a small cycle. Since we initially assumed that $w(c)\ge w(c')$, the second condition in Definition~\ref{def:goodedge} follows and thus, $e$ is a good edge.

  \smallskip\noindent
  \textbf{Case B:} Suppose that $|\ov(e_h)| \ge |\ov(f_{\max})| > |\ov(e_t)|$. We claim that $e=e_h$ is a good edge.
  It follows from Fact~\ref{fact:goodedge-first-condition} that $e$ satisfies the first condition of a good edge.
  Since $|\ov(f)| \ge |\ov(f')|$, it only remains to show the second condition.

  Let $c'$ be the cycle containing the string $v'$. Observe that $e_t$ is an edge in the cycle $c'$ and recall that $f=f_{\max}$ and $|\ov(e_t)| < |\ov(f_{\max})|$. We conclude that $o(c')\le |\ov(e_t)|< |\ov(f)|$, so the second condition in Definition~\ref{def:goodedge} is satisfied and $e$ is good edge.

  \smallskip\noindent
  \textbf{Case C:} Otherwise, since $\max\{|\ov(e_h)|,|\ov(e_t)|\} \ge  |\ov(f_{\max})|$, we have $|\ov(e_t)| \ge |\ov(f_{\max})| > |\ov(e_h)|$. This case is symmetric to the previous one with the roles of $e_t$ and $e_h$ swapped.
  \qedhere
\end{proof}

To summarize, for any arbitrary cycle cover $\C_i$, there exists an edge $e \in \C\setminus \C_i$ such that if we obtain the cycle cover $\C_{i+1}$ from $\C_i$ by performing $\swap(\C_i,e)$, then the total overlap of $\C_{i+1}$ is by at least $\Delta_i$ larger than the total overlap of $\C_i$. This follows because either one of Lemmas~\ref{lem:longcycles} and~\ref{lem:2ring} directly applies or, if that is not the case, Lemma~\ref{lem:centraledge} guarantees the existence of a good edge $e\in \C\setminus \C_i$. For such a good edge, $\swap(\C_i,e)$ provides the claimed increase of the total overlap due to Lemma~\ref{lem:goodedge}.

\section{Final Remarks}\label{sec:finalRemarks}

We have made the first progress since 2005 on the approximation factor of the \GREEDY algorithm, showing that the upper bound of 3.5 by Kaplan and Shafrir~\cite{kaplanshafrir} is not the final answer. In addition, we have also improved the approximation guarantee for the Shortest Superstring problem in general.
Both results follow from our main technical contribution, which is the inequality $o \le n+\alpha\cdot w$ for $\alpha \approx 1.425$.
We get this inequality by proving two incomparable upper bounds on $o$, stated in~\eqref{eq:firstBound} and~\eqref{eq:secondBound},
with the second one being better when large cycles contribute much more to $w$ than small cycles.

Lastly, we want to briefly comment on whether we see potential for further improvements of our results.
We believe that the technique to prove the first bound (largely based on lemmas from previous works) does not offer room for improvement without bringing in substantial new ideas.
On the other hand, our
approach to get the second bound might have more potential for further improvements.
While we do not know specifically how, it would not be surprising to us if arguments in the same spirit could prove inequality $o\leq n +\gamma\cdot \sum_{c\in \mathcal{S}(S)}w(c) +\sum_{c\in \mathcal{L}(S)}w(c)$ for a smaller value of $\gamma$.
However, we also believe that this would make the proof considerably longer and more technical as, for example, more cases about how different small and large cycles interact with each other may have to be considered.
Furthermore, decreasing $\gamma$ slightly does not lead to significantly better upper bounds for \GREEDY or for SSP.
In fact, even for $\gamma = 3$ (compared to the current value of $\approx 3.832$), one would only obtain that $o \le n+1.4\cdot w$, which would imply an upper bound of $3.4$ for \GREEDY (compared to $\approx 3.425$ in Theorem~\ref{thm:greedy}) and of $\approx 2.467$ for the general approximation guarantee of SSP (via Theorem~\ref{thm:mgreedy-to-approx}, which currently gives $\approx 2.475$).

\bibliography{references}
%\newpage
\appendix

\section{Proof of Theorem \ref{thm:mgreedy-to-approx}}
For completeness, we provide a proof of Theorem~\ref{thm:mgreedy-to-approx}. The proof entirely follows the ideas from Theorem 21 in \cite{Mucha07}.

We start by observing that the existence of a $\delta$-approximation algorithm for MaxATSP also implies the existence of a $\delta$-approximation algorithm for MaxATSP path, i.e., the longest Hamiltonian path. This is because we can add one node to our graph which has outgoing and incoming edges to and from all other nodes, and all of these additional edges have profit 0. A TSP tour in this graph corresponds to a Hamiltonian path of equal profit in the original graph, and vice versa.

Given $S$, we compute $\C(S)$ and obtain the set of representative strings $\mathcal{R}$. The string that \MGREEDY generates is simply a concatenation of all these representative strings. Suppose this superstring has length $x$.

If we were not computationally bounded, we could also optimally merge the representative strings instead of naively concatenating them. Optimally merging the representative strings means finding an optimal solution for the SSP instance that has $\mathcal{R}$ as its set of input strings. The following lemma states that optimally merging the representative strings would result in a 2-approximation for the input $S$.

\begin{lemma}\label{lm:optSuperstringForR}
  An optimal superstring for input $\mathcal{R}$ is at most twice as long as an optimal superstring for input $S$.
\end{lemma}
\begin{proof}
  The proof follows the same idea as Lemma~\ref{lem:input-transform}.
  Consider the following superstring for $\mathcal{R}$: For each cycle $c$, we take the single string $s_{c_{0}}$. Let $S'\subseteq S$ be the set of these strings for all cycles. Let $t$ be an optimal superstring for this set $S'$ of strings. Clearly $|t| \le |\OPT(S)|$. Now for each string $s_{c_{0}} \in S'$ replace one occurrence of the string $s_{c_{0}}$ in $t$ by the string $$\pref(s_{c_{0}},s_{c_{1}})\pref(s_{c_{1}},s_{c_{2}})\dots \pref(s_{c_{r-2}},s_{c_{r-1}})\pref(s_{c_{r-1}},s_{c_{0}})s_{c_{0}}\,.$$ This increases the length of $t$ by $\sum_c w(c) \le |\OPT(S)|$ and results in a superstring for $\mathcal{R}$.
\end{proof}

Since we do not know how to compute this optimal solution for $\mathcal{R}$ efficiently, we instead use the $\delta$-approximation algorithm for MaxATSP path on the corresponding overlap graph. Suppose the optimal value (i.e. maximum total overlap) for this MaxATSP path problem is $y$. Then, using Lemma~\ref{lm:optSuperstringForR} and that $x$ is the total length of representative strings, we have $x-y \le 2\cdot|\OPT(S)|$ or, equivalently, $y\ge x - 2\cdot|\OPT(S)|$. Therefore, using the $\delta$-approximation algorithm we obtain a superstring of length at most $x-\delta\cdot y \le x-\delta\cdot (x - 2\cdot|\OPT(S)|) = (1-\delta) \cdot x + 2\delta \cdot |\OPT(S)|$.
Now, the theorem directly follows by using $x \le (2+\alpha)\cdot |\OPT(S)|$.

\end{document}